\documentclass[pdflatex,sn-mathphys-num]{sn-jnl}

\usepackage{graphicx}%
\usepackage{multirow}%
\usepackage{amsmath,amssymb,amsfonts}%
\usepackage{amsthm}%
\usepackage{mathrsfs}%
\usepackage[title]{appendix}%
\usepackage{xcolor}%
\usepackage{textcomp}%
\usepackage{manyfoot}%
\usepackage{booktabs}%
\usepackage{algorithm}%
\usepackage{algorithmicx}%
\usepackage{algpseudocode}%
\usepackage{listings}%
\usepackage{physics}
\usepackage{booktabs}
\usepackage{tikz}
\usetikzlibrary{positioning, shapes.geometric, arrows.meta, snakes}
\usepackage{quantikz}

\newtheorem{theorem}{Theorem}

\newcommand{\ControlledU}[2]{\texttt{C-$\mathrm{U}$}_{#1, #2}}

\theoremstyle{thmstyletwo}%

\theoremstyle{thmstylethree}%

\begin{document}

\title[Article Title]{A Provably Secure Framework for Noise-Aware Delegated Quantum Computation and Storage}


\author[1]{\fnm{Sanidhya} \sur{ Gupta}}\email{sanidhya.gupta@durham.ac.uk}

\author*[1]{\fnm{Ankur} \sur{Raina}}\email{ankur@iiserb.ac.in}


\affil[1]{\orgdiv{Department of Electrical Engineering and Computer Science}, \\  \orgname{Indian Institute of Science Education and Research Bhopal},\\ \orgaddress{ \city{Bhopal}, \state{Madhya Pradesh}, \country{India}}}


\abstract{
As large-scale quantum computers become a reality, they will likely exist as centralized cloud resources accessible to a broad user base. 
Securely delegating private quantum computations to untrusted servers is therefore a foundational challenge. 
This requires rigorous guarantees of privacy (blindness), correctness (completeness), and integrity against malicious actions (verifiability).
This paper presents an integrated architectural framework for noise-aware distributed quantum computation. 
The framework combines three technical components into a unified system: 
(1) a distributed stabilizer-code backbone to encode and store quantum states across multiple server nodes, with security analyzed under non-communication and bounded-collusion assumptions; 
(2) a two-level error-management structure, where each server node can locally handle errors based on its specific noise model; and 
(3) a trap-based verification protocol to detect malicious deviations with probability controlled by a security parameter.
We provide a security analysis showing that, under the stated assumptions, the framework achieves completeness, blindness, and verifiability with respect to the permitted leakage.
Our work provides an architectural blueprint for trustworthy distributed quantum computation under explicitly stated assumptions, paving the way for further development of secure quantum cloud services.
}




\maketitle

\section{Introduction}
\label{sec:introduction}

Distributed quantum computing (DQC) represents a critical pathway to overcoming the scalability limitations of individual quantum processors, leveraging the quantum internet to connect distant devices for collaborative computation \cite{qu_internet_Kimble2008, qu_intrnt_infra_2004}. 
Akin to its classical counterpart \cite{Kshemkalyani_Singhal_2008}, DQC aims to solve complex problems by distributing tasks across a network of quantum computers \cite{sundaram2022distribution, sünkel2024applying, Wu_2023}. 
However, the practical realization of a robust DQC network requires surmounting significant challenges in hardware, algorithmic design, and, most importantly, security.

Progress is being actively pursued on multiple fronts. 
In hardware, developing fault-tolerant quantum repeaters for long-distance communication remains a primary focus \cite{Li2019, Wo2023, Azuma_2023}.
In parallel, theoretical effort has focused on the algorithmic and software layers of DQC, including the development of distributed versions of canonical quantum algorithms such as Grover’s search and Shor’s factoring~\cite{yimsiriwattana2004distributed, xiao2022distributed, jiang2023, qiu2024, zhou2023distributed}, as well as efficient quantum circuit partitioning and entanglement distribution schemes~\cite{andresmartinez2019, andresmartinez2023distributing, avis2023}.
Yet, a critical challenge in this domain is ensuring security under a rigorous, modern cryptographic lens. 
While many protocols address blindness or verifiability \cite{Cai2023}, few are proven secure within a composable framework, which guarantees security when the protocol is used as a component in larger systems \cite{dunjko2014composable, houshmand2018composable}. 
Although these directions have produced important results, a crucial gap remains at the architectural level for a comprehensive framework that unifies computation, secure storage, and multi-level error correction \cite{Ferrari_2023, ferrari2024simulation}. 
Existing works typically study distributed execution, blind or verifiable delegated computation, quantum storage, and quantum error correction as isolated components. 
However, a practical distributed quantum cloud service must integrate these layers into a single operational framework. 

In this paper, we address this architectural and security gap. 
We tackle the following key problem: Suppose Alice, who operates a quantum computer within a larger network, has insufficient local resources. 
She seeks to leverage the computational and storage capabilities of other remote nodes to execute her computations and securely store the resulting states. 
Crucially, Alice requires that these assisting nodes learn nothing about her computation (a property known as \textit{blindness}) and that she can \textit{verify} that the results are correct, even if the remote nodes are malicious.

To address this, we propose and analyze an integrated, noise-aware architecture for distributed quantum computation and storage.
Our contributions are:
\begin{enumerate}
    \item We design an integrated architecture based on a distributed stabilizer-code backbone, which encodes a client’s quantum state across multiple server nodes under explicitly stated non-communication or bounded-collusion assumptions.
    \item We incorporate a two-level error-management structure, allowing each server node to handle local errors independently using noise-aware methods.
    \item We integrate a trap-based verification protocol, enabling the client to detect malicious deviations with probability controlled by the trap security parameter.
    \item We provide a security analysis under explicitly stated assumptions, including non-communication among leaf nodes or bounded collusion below the threshold determined by the level-1 code distance.


\end{enumerate}

We emphasize that the novelty of this work does not lie in introducing a new primitive for blind computation, stabilizer encoding, teleportation-based non-local gates, or trap-based verification in isolation. 
Rather, the contribution is architectural: we integrate these known primitives into a single distributed framework in which quantum storage, delegated computation, local noise-aware error handling, verification, and resource accounting are unified within a single model.
This integration is necessary because these components are not independent in a distributed setting. 
The stabilizer encoding determines the privacy structure of the stored state; the teleportation-based gate layer determines non-local communication cost; the local QEC layer determines how physical node-level noise is handled before it propagates to the global code; and the trap layer determines the detection probability against malicious deviations. 
Treating these components jointly makes the security assumptions, operational costs, and architectural trade-offs explicit.

The structure of this paper is as follows. 
In Section \ref{sec:security_framework}, we introduce the formal security framework and provide the definitions for our security goals.
Section \ref{sec:foundations} reviews the foundational concepts of the stabilizer formalism and establishes the system model upon which our architecture is built.
Section \ref{sec:core_protocols} details our core methodologies, including the distributed encoding algorithms and the procedure for non-local gates. 
In Section \ref{sec:specs}, we outline the full specifications of the architecture, including the trap-based verification protocol and node failure recovery.
In Section \ref{sec:resource_analysis}, we provide the full resource analysis with a concrete example.
Finally,  we present the security analysis, providing proof sketches for our main security claims in Section \ref{sec:security_analysis} and conclude in Section \ref{sec:conclusion}.

\section{Security framework and definitions}
\label{sec:security_framework}

To formally analyze the security of our distributed quantum computation (DQC) architecture, we adopt the \textit{real world vs. ideal world} paradigm, a widely used method in modern cryptography and composable security frameworks \cite{dunjko2014composable, houshmand2018composable, goldreich2004foundations}. 
This approach allows us to define security by comparing our real protocol to a perfect, idealized version of the task.

\subsection{The real world vs. ideal world paradigm}

The security of a protocol is established by comparing two scenarios:
\begin{itemize}
    \item \textbf{The real world:} This setting consists of the actual participants and the protocol they execute. 
    In our case, the participants are the \textit{client} (the master node) and the \textit{server side} (the set of $n$ untrusted leaf nodes). 
    They communicate over the quantum and classical channels described in our architecture, executing all the necessary algorithms. 
    The leaf nodes may deviate from the prescribed protocol and may return incorrect classical outcomes.

    \item \textbf{The ideal world:} This is a hypothetical world containing an incorruptible, trusted third party that executes the task perfectly. 
    We model this party as an \textit{ideal functionality}, denoted by $\mathcal{F}_{\text{DQC}}$, which is a black box that defines the perfect behavior for our task.
\end{itemize}

\textbf{Adversarial model}
For the blindness guarantee, we distinguish between two settings. 
In the baseline setting, the leaf nodes are assumed to be non-communicating during the protocol, except through instructions mediated by the client. 
In the bounded-collusion setting, a subset $T$ of leaf nodes may collude by sharing their local quantum systems and classical transcripts, provided $|T|\le d-1$, where $d$ is the distance of the level-1 $[[n,k,d]]$ stabilizer code. 
Under this condition, erasure of $T$ is correctable by the code, and the reduced state on $T$ is independent of the encoded logical input. 
We do not claim blindness against arbitrary collusion of $d$ or more leaf nodes unless the access structure of the chosen code is analyzed separately.

The core principle of a security proof in this paradigm is to show that for any possible attack by a bounded adversary in the real world, there exists a \textit{simulator} that can interact with the ideal world functionality $\mathcal{F}_{\text{DQC}}$ in such a way that the client's experience is computationally or statistically indistinguishable from her experience in the real world. 
If no algorithm (a ``distinguisher") can tell whether it is interacting with the real protocol or the simulated ideal one, then the real protocol is considered as secure as the ideal functionality itself.

\subsection{The ideal functionality \texorpdfstring{$\mathcal{F}_{\text{DQC}}$}{F\_DQC}}

We formally define the ideal functionality for our delegated quantum computation task as a black box that interacts with a client $\mathcal{C}$ and a server $\mathcal{S}$ according to the following strict rules:

\begin{itemize}
    \item \textbf{Client input:} The client submits her private inputs—a quantum state $\ket{\psi}$ and a classical description of a unitary computation $\mathrm{U}$—directly to $\mathcal{F}_{\text{DQC}}$.
    
    \item \textbf{Permitted leakage to server:} Upon receiving the inputs, $\mathcal{F}_{\text{DQC}}$ computes a \textit{permitted leakage} function, $\lambda(\ket{\psi}, \mathrm{U})$. 
    This function extracts only the information that the protocol is explicitly allowed to reveal.
    For our architecture, this includes the number of physical qubits ($n$), the logical state size ($k$), an upper bound on the gate complexity of $\mathrm{U}$, and the security parameter $k_{trap}$. $\mathcal{F}_{\text{DQC}}$ sends only the result of this function to the server.
    The server learns nothing else about $\ket{\psi}$ or $\mathrm{U}$.
    
    \item \textbf{Server decision:} The server sends a single bit, $b_{cheat} \in \{0, 1\}$, to $\mathcal{F}_{\text{DQC}}$, indicating its intention to act honestly ($0$) or maliciously ($1$).
    
    \item \textbf{Output generation:} If $b_{cheat}=0$, $\mathcal{F}_{\text{DQC}}$ computes $\ket{\psi'} = \mathrm{U}\ket{\psi}$ and delivers it to the client. If $b_{cheat}=1$, $\mathcal{F}_{\text{DQC}}$ discards the computation and delivers a standard, fixed error state $\ket{\text{err}}$ to the client. 
    The server has no ability to influence the final state beyond forcing this specific abort condition.
\end{itemize}
This ideal functionality perfectly captures the desired security goals: the server learns almost nothing, and it cannot produce an incorrect output without the client being notified.

\subsection{Formal security definitions}

A real-world DQC protocol is considered secure if it satisfies the following three fundamental properties, which are standard adaptations of security definitions from the universal composability and abstract cryptography literature \cite{canetti2001universally, maurer2011abstract, Maurer_2012}.
The parameters $\varepsilon_{comp}$, $\varepsilon_{blind}$, and $\varepsilon_{verif}$ are security parameters, which must be negligible functions of a computational security parameter (e.g., the size of the problem).

\subsubsection*{1. Completeness (correctness for honest parties)}
A protocol is complete if, when all parties execute it honestly, the final state of the system is indistinguishable from the ideal outcome. 
Formally, let $\rho_{\text{Client, Server}}^{\text{REAL}}$ be the joint state of the client and the (honest) server at the end of the real protocol. 
Let $\rho_{\text{Client, Server}}^{\text{IDEAL}}$ be the ideal final state, where the client holds the correct output $U\ket{\psi}$ and the server holds the correct final state corresponding to an honest execution (which includes the permitted leakage). 
The protocol is $\varepsilon_{comp}$-complete if:
\begin{equation}
    D(\rho_{\text{Client, Server}}^{\text{REAL}}, \rho_{\text{Client, Server}}^{\text{IDEAL}}) \le \varepsilon_{comp}.
\end{equation}
where $D(\cdot, \cdot)$ is the trace distance. 
This ensures that the entire state of the honest system is negligibly close to perfect.

\subsubsection*{2. Blindness (client privacy)}
A protocol is blind if the server's view is independent of the client's private inputs, given the permitted leakage.
The server's view, $\text{View}_{\mathcal{S}}$, is the complete record of all quantum and classical messages it receives during the protocol execution. 

In the distributed setting considered here, the server's view is interpreted according to the adversarial model. 
In the non-communicating model, blindness is required for the view of each individual leaf node. 
In the bounded-collusion model, blindness is required for the joint view of any colluding subset $T$ satisfying $|T|\le d-1$. 
Collusion of $d$ or more nodes is outside the unconditional blindness guarantee unless the access structure of the chosen code is analyzed separately.

Formally, for any two different sets of client inputs, $(\ket{\psi_1}, \mathrm{U_1})$ and $(\ket{\psi_2}, \mathrm{U_2})$, that produce the same permitted leakage $\lambda(\ket{\psi_1}, \mathrm{U_1}) = \lambda(\ket{\psi_2}, \mathrm{U_2})$, the resulting distributions of the server's view, $\mathcal{V}_1$ and $\mathcal{V}_2$, must be statistically indistinguishable:
\begin{equation}
    D(\mathcal{V}_1, \mathcal{V}_2) \le \varepsilon_{blind}.
\end{equation}
This is proven by constructing a simulator, $\mathcal{SIM}$, which can generate a statistically indistinguishable view using only the leakage information.

\subsubsection*{3. Verifiability (soundness against malicious server)}
A protocol is verifiable if, for any malicious strategy employed by the server, the probability that the protocol finishes without the client aborting, but with the client's output state $\rho_{\text{out}}$ being incorrect, is negligible:

\begin{equation}
    \text{Pr}[\text{Client does not abort} \land D(\rho_{\text{out}}, \ket{\mathrm{U}\psi}\bra{\mathrm{U}\psi}) > \delta] \le \varepsilon_{verif}.
\end{equation}
for any non-negligible distance $\delta$. 
This guarantees that the server cannot trick the client into accepting a wrong answer. 
In our framework, a non-aborted output will be negligibly close to the correct state.

\section{Foundational concepts and system model}
\label{sec:foundations}

Before detailing the core protocols of our architecture, we first establish the system model and review the foundational concepts from quantum information theory upon which our framework is built.

\subsection{System model and network topology}

We consider a distributed quantum computing network composed of two distinct types of participants, defined by their roles in the security model:

\begin{enumerate}
    \item \textbf{The client (master node):} A single entity (Alice) who possesses a private quantum state and a classical description of a quantum computation she wishes to perform. We assume the client has limited quantum memory and processing capabilities, making it infeasible for her to perform the computation locally. The client is considered trusted and acts as the orchestrator of the protocol.

    \item \textbf{The server side (leaf nodes):} A set of $n$ distinct quantum nodes that possess significant quantum storage and computational resources. 
    These nodes are considered powerful but untrusted. 
    Their goal may be to learn about the client's private computation or to return an incorrect result without being detected.
\end{enumerate}

\textbf{Communication and collusion among leaf nodes.}
The baseline architecture assumes that the leaf nodes do not communicate with one another during the protocol, except through instructions mediated by the client. 
This assumption is used to obtain the simplest form of blindness. 
To model more realistic cloud settings in which some nodes may belong to the same provider, we also consider a bounded-collusion extension. 
If the level-1 distributed encoding is an $[[n,k,d]]$ stabilizer code, then any subset $T$ of at most $d-1$ colluding leaf nodes is unauthorized with respect to the encoded quantum data, since erasure of $T$ is correctable by the code. 
Hence, the reduced state available to such a subset is independent of the logical input. 
If $d$ or more nodes collude, the code distance alone no longer guarantees information hiding, and security depends on the access structure of the chosen stabilizer code.

For clarity, we primarily model the network as a star topology, as depicted in Fig.~\ref{fig:system_model}.
In this model, the client acts as the central hub, connected to each of the $n$ server nodes via two types of channels:
\begin{enumerate}
    \item \textbf{Quantum channels:} Physical links capable of generating and distributing entangled Bell pairs between the client and each server node.
    \item \textbf{Authenticated classical channels:} Secure classical links used by the client to send instructions and by the server nodes to return measurement outcomes. Authentication ensures that messages cannot be forged or tampered with by an external party, though it does not prevent a malicious server from sending false information \cite{Vidick_Wehner_2023}.
\end{enumerate}
While the star topology is our primary model for its simplicity and efficiency in centralized protocols, we show in Section \ref{ssec:specs} how our architecture generalizes to any arbitrary connected network topology.

\begin{figure}[h!]
    \centering
    \begin{tikzpicture}[
        node distance=1.5cm and 1cm,
        client/.style={
            rectangle, 
            draw=black!80, 
            fill=blue!20, 
            very thick, 
            rounded corners, 
            minimum height=1.2cm, 
            minimum width=2.8cm, 
            align=center
        },
        server/.style={
            rectangle, 
            draw=black!80, 
            fill=gray!20, 
            very thick, 
            rounded corners, 
            minimum height=1.2cm, 
            minimum width=2.8cm, 
            align=center
        },
        quantum_channel/.style={snake=coil, segment aspect=0, segment length=5pt, line after snake=0pt, very thick, blue!60!black},
        classical_channel/.style={->, very thick, >=stealth, black!70}
    ]

    \node[client] (M) {Client \\ \textbf{Master Node}};

    \node[server, below=4cm of M, xshift=-4cm] (L0) {Server \\ Leaf Node $L_0$};
    \node[server, right=of L0]                  (L1) {Server \\ Leaf Node $L_1$};
    \node[right=of L1]                         (Ldots) {\dots};
    \node[server, right=of Ldots]               (Ln) {Server \\ Leaf Node $L_{n-1}$};

    \foreach \nodename in {L0, L1, Ldots, Ln} {
        \draw[quantum_channel] (M.south) -- ([yshift=2pt]\nodename.north);
        \draw[classical_channel] (M.south) -- ([yshift=-2pt]\nodename.north);
    }
    
    \begin{scope}[xshift=5.25cm, yshift=0.1cm] 
        \node[draw, rounded corners, inner sep=5pt] (legend) {
            \begin{tabular}{@{}ll}
                \tikz{\draw[quantum_channel] (0,0) -- (1,0);} & Quantum channel \\
                \tikz{\draw[classical_channel] (0,0) -- (1,0);} & Authenticated\\ & classical channel \\
            \end{tabular}
        };
    \end{scope}

    \end{tikzpicture}
        \caption{
        \textbf{The Star-topology system model.}
        The architecture is modeled as a star network with two distinct participants.
        \textbf{(1) The client (master node):} A single, trusted entity at the center who orchestrates the computation.
        \textbf{(2) The server side (leaf nodes):} A set of $n$ powerful but untrusted quantum nodes that execute the quantum operations.
        In the baseline blindness model, these leaf nodes are assumed to be non-communicating during the protocol; the bounded-collusion extension is discussed in the system model and security analysis.
        The client connects to each server node via a dedicated quantum channel for Bell-pair distribution and an authenticated classical channel for instructions and measurement outcomes.
    }
    \label{fig:system_model}
\end{figure}

\subsection{Quantum error correction preliminaries}

Quantum states are fragile and susceptible to errors from environmental decoherence and imperfect gate operations, posing a significant challenge to reliable quantum computation. 
Quantum error correction (QEC) is a set of techniques designed to combat these faults by introducing redundancy, a principle that is foundational to all of classical and quantum error correction \cite{nielsen2002quantum}.

The central principle of QEC is to encode a single unit of quantum information, a logical qubit, into the highly entangled state of multiple, redundant physical qubits. 
This encoding maps the two-dimensional logical space into a protected subspace of the much larger $2^n$-dimensional Hilbert space of the $n$ physical qubits. 
This protected subspace is known as the \textit{codespace}. 
The key property of a QEC code is that the effects of different, common physical errors (e.g., a Pauli $\mathrm{X}$, $\mathrm{Y}$, or $\mathrm{Z}$ error on a single qubit) map the codespace to orthogonal, distinguishable subspaces.

By performing a projective measurement that distinguishes between these error subspaces—a process known as \textit{syndrome measurement}—it is possible to identify the error that occurred without disturbing the encoded logical information itself. 
Based on the classical syndrome obtained from this measurement, a corrective operation (typically a Pauli operator) can be applied to reverse the error and restore the state to the original codespace.

\subsection{The stabilizer formalism}

The stabilizer formalism is a powerful mathematical framework for describing and working with a large and important class of QEC codes, known as stabilizer codes \cite{gottesman1997stabilizer}. 
The codes used in our architecture are part of this class.

A stabilizer code is defined by its \textit{stabilizer group}, $S$, which is an Abelian subgroup of the $n$-qubit Pauli group $\mathcal{P}_n$ such that $\mathrm{-I} \notin S$. 
The codespace, $\mathcal{C}$, is the subspace of the $n$-qubit Hilbert space that is simultaneously stabilized (i.e., has an eigenvalue of +1) by every operator in the group:
\begin{equation}
    \mathcal{C} = \{ \ket{\psi} \,|\, s\ket{\psi} = \ket{\psi}, \, \forall s \in S \}.
\end{equation}

A stabilizer group is typically described by a smaller set of independent and commuting operators, $\{s_1, \dots, s_{n-k}\}$, known as the \textit{generators}. 
These generators are sufficient to define the entire group $S = \braket{s_1, \dots, s_{n-k}}$. 
A code described by $n-k$ independent generators encodes $k$ logical qubits into $n$ physical qubits and is denoted as an $[[n,k]]$ code. 

The power of a code is quantified by its \textit{distance}, $d$. 
The distance is the minimum weight (number of non-identity terms) of a Pauli operator that commutes with all stabilizers but is not itself a stabilizer (i.e., a logical operator). 
Such a code can detect any $d-1$ arbitrary physical errors and correct any $t = \lfloor(d-1)/2\rfloor$ arbitrary physical errors \cite{nielsen2002quantum}.

Finally, quantum computation on the encoded data is performed using \textit{logical operators}. 
A logical Pauli operator ($\mathrm{X_{L}}, \mathrm{Y_L}, \mathrm{Z_L}$) is an $n$-qubit Pauli operator that preserves the codespace (commutes with all elements of $S$) but acts non-trivially upon the encoded logical information (is not an element of $S$). 
This formalism provides the complete mathematical foundation for constructing the distributed logical states in our architecture and for defining the logical gate operations that manipulate them.

\section{Core architectural protocols}
\label{sec:core_protocols}

Our architecture is built upon a hierarchy of protocols that manage the quantum state at different levels. 
At the lowest level are the physical primitives that enable non-local operations. 
These primitives are then composed to build the protocols that define the architecture's core functionality. 
At the global level, we use a distributed stabilizer code to encode the logical state across the entire network. 
At the local level, each node employs a second layer of encoding to protect its physical qubits against noise. 
Finally, a set of protocols can be executed to implement logical unitary operations. 


\subsection{The SCST Primitive for non-local controlled gates}
\label{ssec:SCST_primitve}

In our architecture, the fundamental building block for creating entanglement between the client (master node) and a remote server (leaf node) is the Single-Control Single-Target (\texttt{SCST}) gate protocol. 
This protocol, originally detailed by Eisert et al. \cite{eisert2000optimal}, is a teleportation-based primitive that implements a generic controlled-unitary (C-U) operation between a control qubit at the client and a target qubit at a server node.
Its circuit diagram is shown in Fig.~\ref{fig:scst_primitive}.
The mechanism relies on gate teleportation. 
It requires a pre-shared Bell pair ($\ket{\Phi^+}$) between the client and the server node as a quantum resource. 
The protocol proceeds with a sequence of local operations and two bits of classical communication to apply the unitary $U$ to the target qubit, conditioned on the state of the control qubit.

We select the \texttt{SCST} protocol as our communication primitive because it is more resource-efficient for distributed architectures. 
A standard alternative would be to use quantum teleportation to transfer the data qubit from the server to the client, perform the gate locally, and teleport it back. 
This round-trip would consume two Bell pairs. 
In contrast, the \texttt{SCST} protocol achieves the same logical result consuming only one Bell pair, resulting in a $50\%$ reduction in quantum communication overhead.
Since this primitive serves as the foundational component for our distributed encoding (Algorithm~\ref{alg:dist_encoding}) and logical operation (Algorithm~\ref{alg:dist_cu_op}) protocols, we describe it here to ensure the framework is self-contained.


\begin{figure}[h]
    \centering
    \centering
    \begin{center}
    \begin{tikzpicture}
    \node[scale=1]{
    $$
    \begin{quantikz}
        \lstick{\Large $\ket{\psi}_{A}$}  & \ctrl{1} & & & & \gate{\mathrm{\text{\Large Z}}}&\\
        \lstick[2]{\Large $\ket{\phi}^{+}_{A_{e}B_{e}}$} & \targ{} &  \meter{} \wire[d][1]{c}\\
        & &\targ{} &   \ctrl{1} &\gate{\mathrm{\text{\Large H}}} &\meter{}\wire[u][2]{c}\\
        \lstick{\Large $\ket{\psi}_{B}$}  &  && \gate{\mathrm{\text{\Large U}}} \wire[u][1]{c}&&&\\
    \end{quantikz}
    $$
    };
    \end{tikzpicture}
    \end{center}
    \caption{The \texttt{SCST} primitive for non-local controlled-unitary gates, based on the protocol from \cite{eisert2000optimal}. This circuit consumes one pre-shared Bell pair to execute a controlled-unitary gate between a control qubit at node A and a target qubit at node B.}
    \label{fig:scst_primitive}
\end{figure}


\subsection{Level 1 encoding: Distributed stabilizer codes}

The foundation of our architecture is a method for preparing the logical zero state, $\ket{0\dots0}_L$, of an $[[n,k]]$ stabilizer code.
This is achieved by distributing the encoding across $n$ physical qubits, $\{Q_0, Q_1, \dots, Q_{n-1}\}$, where each qubit $Q_j$ is physically held by the corresponding leaf node $L_j$.
Our protocol is a distributed analogue of the standard measurement-based projection method for encoding stabilizer codes, the details of which are reviewed in Appendix~\ref{app:std_encoding}. 
An abstract circuit diagram of our distributed implementation is shown in Fig.~\ref{fig:distributed_encoding_circuit}.

\begin{figure*}
    \centering
    \includegraphics[width=\linewidth]{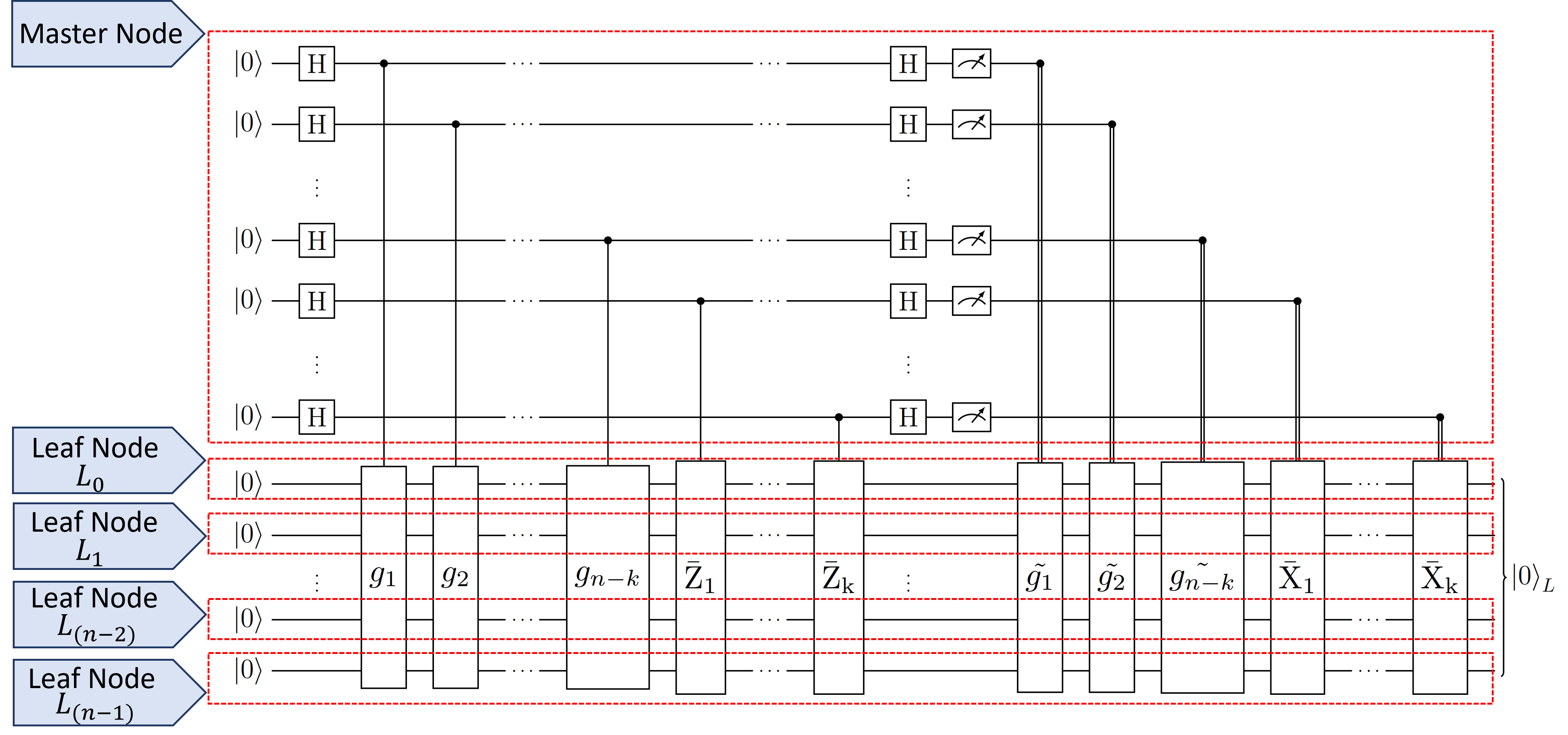}

    \caption{
     \textbf{Abstract circuit for distributed stabilizer encoding.}
    This diagram illustrates the logical flow of Algorithm~\ref{alg:dist_encoding}. The top wires represent the client's (master node) local syndrome qubits, while the bottom wires represent the data qubits at each of the $n$ server (leaf) nodes. Each vertical set of controlled gates corresponds to the enforcement of a single stabilizer generator ($g_i$) or logical operator ($\mathrm{\bar{Z}_j}$) via a series of non-local \texttt{SCST} operations. The number of gates depicted corresponds directly to the Hamming weights of these operators, which determine the total communication cost (see Section~\ref{sec:resource_analysis}).
    }
    \label{fig:distributed_encoding_circuit}
\end{figure*}

The procedure begins with the client (master node) preparing its local qubits in the $\ket{+}$ state. 
These qubits will function as the ancillas for the projective measurements. 
The client then simulates the entangling part of the measurement circuit by executing a series of controlled-Pauli gates between its local qubits and the remote data qubits $\{Q_j\}$ at the server (leaf nodes). 
Each of these non-local gates is implemented using the \texttt{SCST} primitive, as discussed in the previous subsection.

After all entangling operations are complete, the client measures its local qubits in the $\mathrm{X}$-basis. 
This measurement is mathematically equivalent to completing the projective measurements in the standard protocol and yields a classical bit string, the syndrome $m$. 
The purpose of the final protocol steps is to use this syndrome to apply a state correction. 
The key property that the correction operator, $P_{corr}$, must satisfy is that it anti-commutes with every generator corresponding to a `-1' measurement outcome ($m_i=1$) and commutes with all others. 
This deterministically flips the incorrect eigenvalues to `+1', finalizing the state preparation. 
The existence of such a Pauli operator for every possible syndrome is a guaranteed feature of stabilizer codes \cite{nielsen2002quantum}.
This entire distributed procedure is formally described in Algorithm~\ref{alg:dist_encoding}.

As illustrated in Fig.~\ref{fig:distributed_encoding_circuit}, the complexity of the distributed circuit is primarily determined by the algebraic structure of the chosen stabilizer code.
Specifically, the number of non-local operations and thus the quantum communication cost scales linearly with the Hamming weight of the code's generators. 
All non-local communication arises from controlled-Pauli interactions between the client and the leaf nodes, each of which is implemented using the \texttt{SCST} primitive.
Consequently, the total quantum and classical communication cost is determined by the number and weights of the stabilizer generators and logical operators enforced during the encoding procedure. 
A detailed quantitative analysis of the resulting Bell-pair consumption and classical communication overhead is provided in Section~\ref{sec:resource_analysis}.

\begin{algorithm}[H]
    \caption{Distributed stabilizer code encoding.}
    \label{alg:dist_encoding}
    \begin{algorithmic}[1]
        \Require 
        A set of $n$ leaf nodes $\{L_0, \dots, L_{n-1}\}$ and one master node $M$.
        An $[[n,k]]$ stabilizer code defined by its stabilizer group generators $S = \{s_1, \dots, s_{n-k}\}$ and logical $\mathrm{Z}$ operators, $\mathrm{Z_L} = \{\mathrm{\bar{Z}_1}, \dots, \mathrm{\bar{Z}_k}\}$.
        A pre-computed set of correction operators $C$.

        \Ensure The $n$ data qubits $\{Q_j\}$, each held by a leaf node $L_j$, collectively store the logical state $\ket{0\dots0}_L$.
        \Statex 
        \Statex \textbf{Phase 1: Initialization}
        \State The master node $M$ initializes $n$ syndrome qubits $\{q_0, \dots, q_{n-1}\}$ in the state $\ket{+}^{\otimes n}$.
        \State Each leaf node $L_j$ initializes its data qubit $\{Q_j\}$ in the state $\ket{0}$.

        \Statex
        \Statex \textbf{Phase 2: Entangling operations}
        \State Let the set of operators to be enforced be $\mathcal{G} = S \cup \mathrm{Z_L}$. Let $g_i \in \mathcal{G}$ be the operator associated with master syndrome qubit $q_i$.
        \For{each operator $g_i \in \mathcal{G}$}
            \State Decompose $g_i$ into its tensor product of Pauli operators: $g_i = \bigotimes_{j=0}^{n-1} P_j$, where $P_j \in \{\mathrm{I, X, Z, Y}\}$.
            \For{$j=0$ to $n-1$}
                \If{$P_j \neq I$}
                    \State Execute a controlled-$P_j$ gate with master qubit $q_i$ as control and leaf data qubit $Q_j$ as target using the non-local \texttt{SCST} protocol (Fig.~\ref{fig:scst_primitive}).
                \EndIf
            \EndFor
        \EndFor

        \Statex
        \Statex \textbf{Phase 3: Measurement and correction}
        \State The master node measures all its syndrome qubits $\{q_0, \dots, q_{n-1}\}$ in the $\mathrm{X}$-basis, yielding a classical measurement string $m = (m_0, m_1, \dots, m_{n-1})$.
        \State Based on the string $m$, the master node identifies the necessary correction operator $P_{corr} \in C$.
        \State The master node communicates the required corrections to the relevant leaf nodes, which apply them locally to their data qubits $Q_j$.
    \end{algorithmic}
\end{algorithm}

Upon successful completion of the protocol, the joint state of the $n$ data qubits at the server nodes is the desired logical zero state, $\ket{0\dots0}_L$. 
This state is, by construction, the simultaneous `+1' eigenstate of all stabilizer generators and all logical $\mathrm{Z}$ operators. 
It is now a high-fidelity, shared quantum resource, ready to be manipulated by the logical gate operations detailed in the subsequent sections. 
A full algebraic proof of correctness for this protocol is provided in Appendix~\ref{apndx:algo_1_proof}.

\subsection{Level 2 encoding: Local, noise-aware error correction}

A key architectural feature of our framework is a second level of encoding designed to protect against local errors occurring within an individual leaf node.
After the global logical state is established via level 1 encoding, each physical data qubit $Q_i$, held by leaf node $L_i$, is itself a critical resource susceptible to local decoherence. 

To enhance robustness against such faults, we employ a concatenated coding strategy, as illustrated in Fig.~\ref{fig:hierarchical_encoding}.
Specifically, the single data qubit $Q_i$ at each leaf node is further encoded using a local quantum error correction (QEC) code, effectively replacing it with a small, self-contained block of physical qubits at that node. 
This hierarchical approach allows each leaf node to autonomously detect and correct a class of local physical errors. 
The primary technical advantage of this design is the significant reduction in reliance on frequent, resource-intensive, network-wide stabilizer measurements of the level 1 code for every minor physical error event. 
This enhances the overall modularity and efficiency of the architecture by localizing error management.

The choice of this local code involves a critical trade-off between physical qubit overhead, error-correcting capability, and the specific noise profile of the hardware. 
Unlike general-purpose codes (such as the 9-qubit Shor code \cite{shor1995scheme}), custom QEC schemes can be tailored to specific hardware constraints to improve resource efficiency under appropriate noise assumptions. 
We present two such schemes that illustrate different trade-offs between qubit overhead, error coverage, and hardware assumptions.

The role of the local QEC layer is not to provide a universally superior replacement for standard quantum error-correcting codes. 
Rather, it provides a modular mechanism for reducing the effect of dominant local noise processes before they propagate to the level-1 distributed code. 
The benefit, therefore, depends explicitly on the physical noise model of the leaf node. 
If the local physical error rate is $p$, and the chosen local scheme corrects all single errors in a protected set $\mathcal{E}_{\mathrm{corr}}$, then the leading uncorrected contribution from that protected error class is suppressed from first order to second order, i.e., from $O(p)$ to $O(p^2)$. 
Errors outside the protected set are either uncorrected or flagged, depending on the scheme. 
Thus, the local QEC layer improves reliability primarily when the dominant hardware noise lies within $\mathcal{E}_{\mathrm{corr}}$.

\begin{figure}[h]
    \centering
          \resizebox{\textwidth}{!}{%
  \begin{tikzpicture}[
      node distance = 2cm and 3cm,
      every node/.style = {font=\footnotesize, align=center},
      arrow/.style      = {thick,-{Stealth}},
      level1/.style     = {draw,fill=orange!30,thick,minimum width=3cm,
                           minimum height=1cm,rounded corners},
      level2/.style     = {draw,fill=green!30,thick,minimum width=3cm,
                           minimum height=1cm,rounded corners},
      dashedbox/.style = {draw,dashed,rounded corners,inner sep=6pt}
  ]

  \node[draw,thick,fill=magenta!20,rounded corners,
        minimum width=3cm] (init) {Initial state\\$|\psi\rangle$};

  \node[level1,right=4cm of init] (Q0) {Physical qubit $Q_0$\\at leaf node $L_0$};
  \node[level1,below=0.8cm of Q0] (Q1) {Physical qubit $Q_1$\\at leaf node $L_1$};
  \node[below=0.8cm of Q1,font=\Large] (Qdots) {$\vdots$};
  \node[level1,below=0.8cm of Qdots] (Qn) {Physical qubit $Q_{\,n-1}$\\at leaf node $L_{n-1}$};

  \node[above=0.8cm of Q0] (L1Label) {\textbf{Level 1: Global encoding}\\
           \footnotesize(Distributed stabilizer code)};
  \node[dashedbox,fit={(Q0)(Q1)(Qdots)(Qn)(L1Label)}] {};

  \node[level2,right=4cm of Q0] (P0) {Local QEC block\\$\{Q_{00},Q_{01},\ldots\}$};
  \node[level2,below=0.8cm of P0] (P1) {Local QEC block\\$\{Q_{10},Q_{11},\ldots\}$};
  \node[below=0.8cm of P1,font=\Large] (Pdots) {$\vdots$};
  \node[level2,below=0.8cm of Pdots] (Pn) {Local QEC block\\$\{Q_{\,(n-1)0},\ldots\}$};

  \node[above=0.8cm of P0] (L2Label) {\textbf{Level 2: Local encoding}\\
           \footnotesize(Concatenated QEC)};
  \node[dashedbox,fit={(P0)(P1)(Pdots)(Pn)(L2Label)}] {};

  \node[draw,thick,rounded corners=4pt,fill=gray!20,
        text width=8cm,align=center,
        above=3cm of $(init)!0.5!(P0)$] (title)
        {\textbf{Hierarchical encoding architecture}};

  \draw[arrow] (init.east) -- ++(1,0) |- (Q0.west)
               node[midway,above,font=\scriptsize]{Algorithm 1};
  \foreach \dest in {Q1,Qn}
    \draw[arrow] (init.east) -- ++(1,0) |- (\dest.west);

  \foreach \i/\j in {Q0/P0,Q1/P1,Qn/Pn}
    \draw[arrow] (\i.east) -- node[midway,above,font=\scriptsize]{Local QEC} (\j.west);

  \end{tikzpicture}%
  }
    \caption{
        \textbf{The hierarchical two-level encoding architecture.} 
        This diagram illustrates the concatenated coding strategy. 
        \textbf{Level 1 (global encoding):} The initial state $\ket{\psi}$ is first encoded into a logical state distributed across $n$ physical qubits $\{Q_0, \dots, Q_{n-1}\}$ using Algorithm~\ref{alg:dist_encoding}, where each qubit $Q_i$ resides at a separate server node. 
        \textbf{Level 2 (local encoding):} Each of these individual physical qubits $Q_i$ is then further encoded into a block of new physical qubits $\{Q_{i0}, Q_{i1}, \dots\}$ using a local QEC scheme. This allows each node to handle local errors autonomously.
    }
    \label{fig:hierarchical_encoding}
\end{figure}

\subsubsection{Method 1: A 4-qubit code for arbitrary single-qubit errors}
For scenarios where ancillary qubits can be prepared and maintained with very high fidelity (a ``perfect ancilla" model), we designed a procedural QEC scheme (depicted in Fig.~\ref{fig:qec_method_1}) that uses three such perfect ancillas to protect one data qubit.
Our proposed 4-qubit scheme exploits the availability of high-fidelity ancilla qubits to achieve deterministic correction of any single Pauli error acting on the data qubit, while using only three ancillary qubits. 
This design illustrates a deliberate trade-off that favors reduced qubit overhead when reliable ancilla preparation is available. 
Such assumptions are reasonable in architectures where ancilla qubits can be prepared offline and verified independently of the data qubits.
A full algebraic correctness proof for this scheme is provided in Appendix~\ref{app:qec_proofs}.

A standard alternative in this setting would be the 3-qubit repetition code \cite{nielsen2002quantum}. 
However, the 3-qubit repetition code can correct only one type of Pauli error (either $X$ or $Z$, depending on the encoding) and offers no protection against the complementary error type. 
In contrast, correcting arbitrary single-qubit Pauli errors ($X, Y, Z$) using a conventional stabilizer code typically requires at least five physical qubits, as in the $[[5,1,3]]$ code \cite{laflamme1996}, and relies on additional ancillary qubits and measurements for syndrome extraction.

\begin{figure}[h]

    \centering
	\begin{center}
		\begin{tikzpicture}
			\node[scale=0.75]{
\begin{quantikz}
\lstick{\LARGE $\ket{\psi}$} & \ctrl{2} & \gate{\mathrm{\text{\Large H}}} & \ctrl{1} & \gate[1]{\mathrm{\text{\Large Error}}}\gategroup[4,steps=1,style={dashed,rounded
corners,fill=blue!20, inner xsep=2pt},background,label style={label position=below,anchor=north,yshift=-0.2cm}]{} & \ctrl{1} & \targ{}   & \gate{\mathrm{\text{\Large H}}} & \ctrl{2} & \targ{}   & \targ{} & \rstick{\LARGE $\ket{\psi}$}\\
\lstick{\LARGE $\ket{0}$}        &          &          & \targ{}  &                 & \targ{}  & \ctrl{-1} &          &          &           & &\\
\lstick{\LARGE $\ket{0}$}        & \targ{}  & \gate{\mathrm{\text{\Large H}}} & \ctrl{1} &                 & \ctrl{1} & \targ{}   & \gate{\mathrm{\text{\Large H}}} & \targ{}  & \ctrl{-2} & \ctrl{-2} &\\
\lstick{\LARGE $\ket{0}$}        &          &          & \targ{}  &                 & \targ{}  & \ctrl{-1} &          &          &           & \ctrl{-3} &
\end{quantikz}
 };
\end{tikzpicture}
\end{center}
\caption{\textbf{Circuit for the 4-qubit local QEC scheme (Method 1).}
        This procedural circuit is designed to protect one data qubit using three ancillary qubits. It can deterministically detect and correct any single Pauli error ($\mathrm{X}$, $\mathrm{Y}$ or $\mathrm{Z}$) that occurs on the data qubit. This scheme operates under the idealized assumption that the ancillary qubits themselves are error-free.}
\label{fig:qec_method_1}
\end{figure}

\subsubsection{Method 2: A 6-qubit Scheme Optimized for biased noise}

Many leading physical quantum platforms, such as superconducting cat qubits or Rydberg atoms, exhibit biased noise, where certain error types (e.g., $X$ or $Y$) are orders of magnitude more probable than others (e.g., $Z$) \cite{guillaud2019repetition, cong2022hardware}. 
Using a symmetric code in this regime may be resource-inefficient. 
To leverage this physical reality, we designed a second QEC scheme (Fig.~\ref{fig:qec_method_2}) tailored for such a noise model where $\mathrm{X}$ and $\mathrm{Y}$ errors are dominant.
This 6-qubit protocol makes a deliberate engineering trade-off: it provides deterministic correction for the most probable errors ($\mathrm{X}$ and $\mathrm{Y}$) while treating less likely $\mathrm{Z}$ errors as flagged events rather than fully correctable faults.
Specifically, it offers:
\begin{itemize}
    \item \textbf{Full correction for dominant errors:} The protocol detects and corrects any single $\mathrm{X}$ or $\mathrm{Y}$ error on any of the six involved qubits.
    \item \textbf{Flagging for less likely errors:} A single $\mathrm{Z}$ error produces a unique syndrome (\texttt{00100}) that acts as a ``flag".
\end{itemize}

To make the noise assumption explicit, let $p_{XY}$ denote the total probability of a dominant $\mathrm{X}$- or $\mathrm{Y}$-type fault on a local block during one correction cycle, and let $p_Z$ denote the probability of a $\mathrm{Z}$-type fault. 
The 6-qubit scheme is designed for the biased regime
\[
p_Z \ll p_{XY}.
\]
Because all single $\mathrm{X}$- and $\mathrm{Y}$-type faults in the protected set are deterministically corrected, the leading uncorrected contribution from these dominant faults is pushed to second order, $O(p_{XY}^2)$. 
In contrast, single $Z$-type faults are not corrected; they produce the flag syndrome \texttt{00100} and therefore contribute as flagged events of order $O(p_Z)$. 
Consequently, this scheme is advantageous when $p_Z$ is sufficiently small and when the higher-level protocol can treat a $Z$-flag event as an erasure, abort condition, or request for re-encoding.

When a node measures this flag, it knows that an error outside the deterministically correctable set has occurred and can report this failure to the client's higher-level fault-tolerance routines. 
This noise-aware approach provides protection against the dominant local errors with fewer physical qubits than a general-purpose code designed to correct all single-qubit Pauli errors equally.
The correctness of this scheme with respect to its biased-noise design goal is formally proven in Appendix~\ref{app:qec_proofs}.

The standard approach for full single-qubit error correction is the $[[9,1,3]]$ Shor code, which requires 9 physical qubits.
By tailoring the code to the biased noise model, our scheme protects against the dominant errors using only 6 qubits (1 data + 5 ancilla). 

This represents a $33\%$ reduction in local qubit count relative to the 9-qubit Shor code. 
However, this reduction is obtained under a biased-noise assumption and with weaker treatment of rare $Z$-type faults: the scheme corrects dominant single $X/Y$ faults but flags, rather than corrects, single $Z$ faults. 
Thus, the scheme should be viewed as a noise-tailored engineering trade-off, not as a universally more efficient replacement for a distance-3 code.

\begin{figure}
    \centering
\begin{center}
\begin{tikzpicture}
\node[scale=0.75]{
$$
\begin{quantikz}
\lstick{\LARGE $\ket{\psi}$} & \ctrl{3} & \gate{\mathrm{\text{\Large H}}} & \ctrl{1} & \ctrl{2} & \gate[6, disable auto height]{\mathrm{\text{\LARGE \verticaltext{Error}}}}\gategroup[6,steps=1,style={dashed,rounded
corners,fill=blue!20, inner
xsep=2pt},background,label style={label
position=below,anchor=north,yshift=-0.2cm}]{{\sc
}} & \ctrl{1} & \ctrl{2} & \targ{}   & \gate{\mathrm{\text{\Large H}}}  & \ctrl{3} & \targ{}   &\targ{}   &\targ{}   &\targ{}   & \rstick{\LARGE $\ket{\psi}$}\\
\lstick{\LARGE $\ket{0}$}    & &   & \targ{}  & &                 & \targ{}  && \ctrl{-1} & &         &          &           & & & \\
\lstick{\LARGE $\ket{0}$}    & &   &  & \targ{}  &                & & \targ{}  & \ctrl{-2} & &         &          &           & & & \\
\lstick{\LARGE $\ket{0}$} & \targ{} & \gate{\mathrm{\text{\Large H}}} & \ctrl{1} & \ctrl{2} &  & \ctrl{1} & \ctrl{2} & \targ{}   & \gate{\mathrm{\text{\Large H}}}  & \targ{} & \ctrl{-3}   & \ctrl{-3}   & \ctrl{-3}   & \ctrl{-3}   &\\
\lstick{\LARGE $\ket{0}$}    & &   & \targ{}  & &                 & \targ{}  && \ctrl{-1} & &        &          &            \ctrl{-4} & & \ctrl{-4} &\\
\lstick{\LARGE $\ket{0}$}    & &   &  & \targ{}  &                & & \targ{}  & \ctrl{-2} & &        &          &            \ctrl{-5} & \ctrl{-5} & & 
\end{quantikz}
$$
};
\end{tikzpicture}
\end{center}
    \caption{\textbf{Circuit for the 6-qubit biased-noise QEC scheme (Method 2).}
        This procedural circuit is engineered for a biased noise model where Pauli $\mathrm{X}$ and $\mathrm{Y}$ errors are the dominant faults. 
        It uses five ancillary qubits to protect one data qubit and provides deterministic correction for any single $\mathrm{X}$ or $\mathrm{Y}$ error on any of the six qubits. 
        In contrast, a single Pauli $\mathrm{Z}$ error, being less probable in this model, produces a unique, unambiguous ``flag" syndrome (\texttt{00100}) that signals a non-correctable error event.}
    \label{fig:qec_method_2}
\end{figure}

Should a node not exhibit a known noise bias, or if universal error protection is required, our architecture's flexibility allows for the incorporation of standard, general-purpose codes. 
This modularity extends to the network level: different QEC methods can be deployed on individual leaf nodes to match their specific hardware requirements. 

Thus, the role of the level-2 layer is not to prescribe a single universally optimal local code, but to allow the local protection mechanism to be selected according to the noise profile and resource constraints of each node.

\subsection{Executing gates between non-adjacent nodes}
\label{ssec:algo2}

While the \texttt{SCST} primitive handles gates between directly connected nodes, a general architecture requires a method to perform two-qubit gates between arbitrary, non-adjacent nodes. 
Our architecture achieves this with the protocol detailed in Algorithm~\ref{alg:dist_cu_op}, which uses an intermediary node as a quantum relay.
The scenario, depicted in Fig.~\ref{fig:non_local_gate_protocol}, involves three nodes: a control Node A, an intermediary Node B, and a target Node C. 
The protocol enables a logical interaction between A and C by composing two distinct phases.

\begin{figure}[h!]
    \centering
    \resizebox{\textwidth}{!}{%
    \begin{tikzpicture}[
        node distance=3cm,
        client/.style={
            rectangle, 
            draw=black!80, 
            fill=blue!20, 
            very thick, 
            rounded corners, 
            minimum height=1.2cm, 
            minimum width=2.8cm, 
            align=center
        },
        server/.style={
            rectangle, 
            draw=black!80, 
            fill=gray!20, 
            very thick, 
            rounded corners, 
            minimum height=1.2cm, 
            minimum width=2.8cm, 
            align=center
        },
        quantum_channel/.style={snake=coil, segment aspect=0, segment length=5pt, line after snake=0pt, very thick, blue!60!black},
        classical_channel/.style={->, very thick, >=stealth, black!70},
        phase_label/.style={font=\bfseries\large, anchor=west},
        result_label/.style={font=\itshape, align=center, text width=4cm}
    ]

    \node[phase_label] (P1_label) at (2, 2) {Phase 1: Entanglement swapping};
    
    \node[server] (A1) {Node A (Control)};
    \node[server, right=of A1] (B1) {Node B (Intermediary)};
    \node[server, right=of B1] (C1) {Node C (Target)};
    
    \node[below=0.3cm of B1, circle, draw, inner sep=1pt, fill=white] (BSM) {BSM};
    
    \draw[quantum_channel] (A1.east) to  (B1.west);
    \draw[classical_channel] (A1.east) to  (B1.west);
    \draw[quantum_channel] (B1.east) to (C1.west);
    \draw[classical_channel] (B1.east) to (C1.west);
    
    \draw[classical_channel] (BSM) -> node[midway, below, font=\small] {$(m_1, m_2)$} (C1.south);
    
    \node[result_label, text width=7cm, below=1.2cm of B1] (Result1) {Establishes direct link: $\ket{\Phi^+}_{A_e, C_e}$};
    
    \node[phase_label] (P2_label) at (2, -3.25) {Phase 2: Gate teleportation (\texttt{SCST})};

    \node[server, below=4.5cm of A1] (A2) {Node A \\ $\ket{\psi}_A$};
    \node[server, below=4.5cm of C1] (C2) {Node C \\ $\ket{\psi}_C$};
    
    \draw[quantum_channel] (A2.east) to node[midway, above, font=\small] {Resource: $\ket{\Phi^+}_{A_e, C_e}$} (C2.west);
    
    \draw[classical_channel] (A2.east) to[bend left=25] node[midway, above, font=\small] {$m_{A_e}$} (C2.west);
    \draw[classical_channel] (C2.west) to[bend left=25] node[midway, below, font=\small] {$m_{C_e}$} (A2.east);

    \node[result_label, text width=7cm, below=1.2cm of A2, xshift=6.5cm] (Result2) {Final state: $\ControlledU{A}{C}(\ket{\psi}_A \otimes \ket{\psi}_C)$};

    \end{tikzpicture}
    }
    \caption{
        \textbf{Mechanism of Algorithm~\ref{alg:dist_cu_op}.}
        The protocol operates in two distinct phases. 
        \textbf{Phase 1 (Entanglement swapping):} Two local Bell pairs are created (A-B and B-C). A Bell state measurement (BSM) at the intermediary node B consumes these pairs to create a single, long-distance entangled Bell pair between the non-adjacent nodes A and C.
        \textbf{Phase 2 (Gate teleportation):} This long-distance pair is then used as a resource for the \texttt{SCST} primitive. A sequence of local measurements and classical communication between A and C teleports the conditional logic of the gate, completing the non-local operation.
    }
    \label{fig:non_local_gate_protocol}
\end{figure}

\subsubsection*{Phase 1: Entanglement swapping}
The protocol first establishes an entangled Bell pair between the non-adjacent nodes A and C. 
We leverage the well-established technique of entanglement swapping \cite{zukowski1993event}. 
In this standard procedure, the intermediary Node B performs a Bell-state measurement (BSM) on the ancillary qubits from two local Bell pairs (A-B and B-C). 
The measurement outcome allows Node C to apply a local Pauli correction, effectively stitching the entanglement to create a single, long-distance resource $\ket{\Phi^+}_{A_e C_e}$ shared between A and C.

\subsubsection*{Phase 2: Gate teleportation}
With the entangled link established, the protocol executes the gate using the \texttt{SCST} primitive (described in Section~\ref{ssec:SCST_primitve}). 
The long-distance Bell pair from Phase 1 serves as the quantum resource. 
The client coordinates the \texttt{SCST} sequence over this link to apply the non-local controlled-$U$ operation.

\begin{algorithm}[H]
    \caption{Distributed controlled-$\mathrm{U}$ operation between non-adjacent nodes.}
    \label{alg:dist_cu_op}
    \begin{algorithmic}[1]
        \Require 
        Control node A with data qubit $\ket{\psi}_A$ and ancilla $\ket{0}_{A_e}$.
        Intermediate node B with ancillas $\ket{0}_{B_{e1}}, \ket{0}_{B_{e2}}$.
        Target node C with data qubit $\ket{\psi}_C$ and ancilla $\ket{0}_{C_e}$.
        Unitary $\mathrm{U}$.
        \Ensure The final state is $\ControlledU{A}{C}(\ket{\psi}_{A} \otimes \ket{\psi}_{C})$.
        
        \Statex \textbf{Phase 1: Entanglement swapping}
        \State Create local Bell pairs $\ket{\Phi^+}_{A_e B_{e1}}$ and $\ket{\Phi^+}_{C_e B_{e2}}$.
        \State Node B performs a Bell state measurement on $(B_{e1}, B_{e2})$ and sends the result to Node C.
        \State Node C applies the Pauli correction to ancilla $C_e$.
        \Statex \Comment{\texttt{A direct entangled link $\ket{\Phi^+}_{A_e C_e}$ is now established.}}

        \Statex \textbf{Phase 2: Gate teleportation}
        \State Execute the \texttt{SCST} protocol using the link $\ket{\Phi^+}_{A_e C_e}$ as the resource.
    \end{algorithmic}
\end{algorithm}

This protocol is provably correct and extends the effective connectivity of the architecture beyond immediate neighbours.
A full algebraic correctness proof is provided in Appendix~\ref{apndx:algo_2_proof}.

\subsection{Distributed logical operations for universal computation}
\label{ssec:distributed_ops}

Once a logical state is encoded and distributed across the server nodes, the architecture supports universal quantum computation by executing arbitrary logical unitaries on the encoded data. 
Our framework integrates standard quantum circuit compilation techniques into a distributed setting coordinated by the client.
The execution of a logical unitary proceeds as a three-stage pipeline, orchestrated by the client’s classical controller (see Fig.~\ref{fig:computation_pipeline}). 
First, the client decomposes the target unitary into a sequence of elementary gates drawn from a universal gate set, such as $\{\mathrm{H, S, CNOT, T}\}$, using standard quantum circuit synthesis techniques \cite{nielsen2002quantum}. 
Each elementary gate is then mapped to its corresponding logical operation according to the chosen stabilizer code. 
The specific implementation method - transversal execution, logical gate decomposition, or magic state injection \cite{bravyi2005universal} - is selected by the client’s controller based on the algebraic properties of the code. 
Finally, the client coordinates the execution of each logical gate using the distributed protocols supported by the architecture, including transversal local operations at the server nodes and non-local gate teleportation protocols for entangling logical operations.

A comprehensive technical detail regarding the logical gate implementations and the associated control logic is provided in Appendix~\ref{apndx:U_description}.
The high-level control flow governing this execution pipeline is formally encapsulated in Algorithm~\ref{alg:logical_unitary_generalized}.
By composing these standard techniques, the architecture provides a complete and modular framework for executing arbitrary quantum algorithms on distributed, encoded quantum data.

\begin{figure}[h]
    \centering
    \begin{tikzpicture}[
        scale=1,
        node distance=1.2cm and 2cm,
        every node/.style={transform shape},
        process/.style={rectangle, draw=black!80, fill=blue!10, very thick, rounded corners, minimum height=1cm, minimum width=4cm, align=center},
        decision/.style={diamond, draw=black!80, fill=green!10, very thick, aspect=2, align=center, font=\small},
        data/.style={trapezium, trapezium left angle=70, trapezium right angle=110, draw=black!80, fill=orange!10, thick, align=center},
        arrow/.style={->, thick, >=stealth}]

    \node[data] (input) {\textbf{Client input}\\ Target unitary $U$};

    \node[process, below=of input] (stage1)
    {\textbf{Stage 1: Classical compilation}\\
    Decompose $U$ into gate sequence $G$};
    
    \node[process, below=of stage1] (stage2)
    {\textbf{Stage 2: Logical gate mapping}\\
    Map gates in $G$ to logical operators};
    
    \node[process, below=of stage2] (stage3)
    {\textbf{Stage 3: Distributed execution}\\
    Client dispatches protocols to server nodes};
    
    \draw[arrow] (input) -- (stage1);
    \draw[arrow] (stage1) -- (stage2);
    \draw[arrow] (stage2) -- (stage3);

    \end{tikzpicture}
    \caption{\textbf{Universal computation pipeline.}
    The client’s controller translates an abstract unitary $U$ into a distributed execution plan.
    The process proceeds from classical compilation (Stage 1), to logical gate mapping under the chosen stabilizer code (Stage 2), and finally to distributed execution (Stage 3) via the architecture’s supported protocols. The associated control logic and implementation choices are formalized in Algorithm~\ref{alg:logical_unitary_generalized}, with additional details provided in Appendix~\ref{apndx:U_description}.}
    \label{fig:computation_pipeline}
\end{figure}


\begin{algorithm}
    \caption{Generalized distributed execution of a logical unitary sequence.}
    \label{alg:logical_unitary_generalized}
    \begin{algorithmic}[1]
        \Require 
        From the client's classical compiler: a gate sequence $G = (g_1, g_2, \dots, g_m)$ from the universal set $\{\mathrm{H, S, CNOT, T}\}$.
        A classical description of the chosen $[[n,k]]$ stabilizer code's properties (e.g., transversality of gates).
        The system, initialized in the state $\ket{0\dots0}_L$.

        \Ensure The final state of the system is $\mathrm{U_L}\ket{0\dots0}_L$, where $\mathrm{U_L} = \prod_{i=1}^m g_{i,L}$.

        \Statex
        \For{each elementary gate $g_i$ in the sequence $G$}
            \Statex \Comment{ \texttt{The client's controller determines the optimal implementation method for $g_{i,L}$ based on the code's properties.}}
            
            \If{$g_i$ is a Clifford gate ($\mathrm{H, S, CNOT}$)}
                \If{the gate $g_{i,L}$ is transversal for the chosen code}
                    \State \textbf{Execute transversal operation:} The client instructs the relevant server nodes to apply the corresponding physical gate ($\mathrm{H, S, CNOT}$) locally to the physical data qubits.
                \Else
                    \Statex \Comment{ \texttt{Execute non-transversal Clifford gate via gadgetry or decomposition.}}
                    \If{$g_i$ is a single-qubit gate ($\mathrm{H, S}$)}
                         \State Execute a known, pre-compiled sequence of other available logical gates that synthesizes the desired gate.
                    \ElsIf{$g_i$ is a two-qubit $\mathrm{CNOT}$ gate}
                        \State \textbf{Execute via $\mathrm{CZ}$ decomposition:} Execute the sequence $\mathrm{H_L} \cdot \mathrm{CZ_L} \cdot \mathrm{H_L}$ on the target logical qubit, where the non-local $\mathrm{CZ_L}$ is implemented using \textbf{Algorithm \ref{alg:dist_cu_op}}.
                    \EndIf
                \EndIf
                
            \ElsIf{$g_i$ is a non-Clifford gate ($\mathrm{T}$)}
                \Statex \Comment{ \texttt{Execute via magic state injection protocol.}}
                \State A high-fidelity logical magic state, $\mathrm{\ket{T}_L}$, is provided by a magic state factory (e.g., the client).
                \State The client coordinates a logical $\mathrm{CNOT}$ gate with the data logical qubit as control and the $\mathrm{\ket{T}_L}$ ancilla as target.
                \State The client instructs the server to perform a logical measurement on the \textit{ancilla qubit} in the $\mathrm{X}$-basis.
                \State The classical outcome $m \in \{0,1\}$ is communicated from the server to the client.
                \State The client instructs the server to apply a final Clifford correction (e.g., $(\mathrm{S_L}^\dagger)^m$) to the data logical qubit.
                \Statex \Comment{\texttt{The data qubit is now in the state $\mathrm{T_L}\ket{\psi}_L$. This consumes one magic state.}}
            \EndIf
        \EndFor
    \end{algorithmic}
\end{algorithm}

\section{Architectural specifications and security protocols}
\label{sec:specs}

In this section, we detail the operational specifications of our architecture and describe the protocols that enable it to meet the formal security definitions of completeness, blindness, and verifiability established in Section~\ref{sec:security_framework}.

\subsection{Distributed quantum storage and computation}

The primary function of the architecture is to provide a distributed quantum state storage and computation service.
The method for constructing a distributed stabilizer code (Algorithm \ref{alg:dist_encoding}) and executing logical operations (Algorithm \ref{alg:logical_unitary_generalized}) enables the creation and manipulation of a logical state that does not physically reside at the client's location. 
This approach offers several benefits:
\begin{itemize}
    \item \textbf{Decentralized load:} It allows a computationally limited client to leverage the superior storage and processing capabilities of the server side, consisting of the leaf nodes.
    \item \textbf{Enhanced security:} By storing the quantum state non-locally across multiple leaf nodes, the system gains a threshold form of information hiding. For a level-1 $[[n,k,d]]$ stabilizer encoding, any colluding subset $T$ of at most $d-1$ leaf nodes is unauthorized with respect to the encoded quantum data: the reduced state available to such a subset is independent of the logical input. Thus, compromising a small number of leaf nodes is insufficient to reveal the client's encoded quantum state. Larger colluding subsets may access nontrivial logical information depending on the access structure of the specific stabilizer code, and are not covered by the unconditional blindness claim made here.

\end{itemize}

\subsection{Protocol for trap-based verification}
\label{ssec:verification}

To ensure security against a malicious server that may deviate from the prescribed protocols, our architecture must be verifiable. A malicious server might apply incorrect gates, use faulty ancillary states, or provide false measurement outcomes. To detect such deviations, we integrate a formal trap-based verification scheme, a standard and powerful technique in delegated quantum computation.

The protocol, detailed in Algorithm~\ref{alg:verification}, involves the client secretly embedding \textit{trap} qubits at random locations within the main computation. These traps have known initial states and are expected to undergo simple, predictable transformations (e.g., the Identity operation). By measuring the final state of these traps, the client can perform a statistical test to verify the server's honest behavior.

\begin{algorithm}[H]
    \caption{Trap-based verification for distributed protocols.}
    \label{alg:verification}
    \begin{algorithmic}[1]
        \Require 
        The client's private computational task.
        A security parameter $k_{traps}$.
        A set of $n_{data}$ physical qubits is required for the computation.
        A total of $N = n_{data} + k_{traps}$ physical leaf nodes.

        \Ensure 
        If the server is honest, the client obtains the correct final state.
        If the server is malicious, the client aborts with probability at least $1-\varepsilon$, where $\varepsilon$ is controlled by the number of traps $k_{\mathrm{trap}}$.

        \Statex
        \Statex \textbf{Phase 1: Client-side trap preparation}
        \State The client secretly partitions the $N$ leaf nodes into two sets: a set of $n_{data}$ nodes for the computation and a set of $k_{traps}$ nodes designated as traps.
        \State For each trap node, the client prepares its qubit in a state randomly chosen from the BB84 set $\mathcal{B} = \{\ket{0}, \ket{1}, \ket{+}, \ket{-}\}$.
        \State The client secretly records the initial state $\ket{s_t}$ and corresponding measurement basis $B_t$ ($\mathrm{Z}$-basis for $\{\ket{0}, \ket{1}\}$, $\mathrm{X}$-basis for $\{\ket{+}, \ket{-}\}$) for each trap $t$.

        \Statex
        \Statex \textbf{Phase 2: Blinded protocol execution with traps}
        \State The client executes the main distributed protocol over the full set of $N$ nodes. The server does not know which nodes are traps.
        \State For any operation involving a qubit that the client knows is a trap, she provides instructions for the logical Identity gate.
        
        \Statex
        \Statex \textbf{Phase 3: Remote verification check}
        \State After the main protocol concludes, for each of the $k_{traps}$ trap locations, the client instructs the corresponding leaf node to measure its qubit in the secret basis $B_t$ and to return the classical measurement outcome.
        \For{each returned classical measurement outcome}
            \If{the outcome does not match the known initial state of the trap}
                \State \textbf{ABORT.} The client discards the computational result and outputs the designated error state $\ket{\text{err}}$.
            \EndIf
        \EndFor
        \State If all $k_{traps}$ traps pass the check, the client accepts the result of the main computation as valid.
    \end{algorithmic}
\end{algorithm}


The random placement of these traps means that any non-trivial deviation by the server from the prescribed protocol has a probability of disturbing at least one trap. 
The number of traps directly controls the soundness error. 
Suppose a malicious deviation affects $d_{\mathrm{mal}}$ distinct physical qubit pathways among $N$ total pathways. 
Since the trap locations are hidden from the server and chosen uniformly at random, the probability that the malicious deviation avoids all trap locations is bounded by
\[
P_{\mathrm{undetected}}
\leq
\left(1-\frac{k_{\mathrm{trap}}}{N}\right)^{d_{\mathrm{mal}}}
\leq
\exp\left(-\frac{d_{\mathrm{mal}}k_{\mathrm{trap}}}{N}\right).
\]
Thus, increasing $k_{\mathrm{trap}}$ decreases the probability of undetected cheating exponentially in $d_{\mathrm{mal}}.k_{\mathrm{trap}}/N$. 
Equivalently, under this simplified independent-pathway model, a target soundness error $\varepsilon$ can be achieved by choosing
\[
k_{\mathrm{trap}} \approx \frac{N}{d_{\mathrm{mal}}}\ln(1/\varepsilon).
\]
This improved detection probability comes with a direct resource cost, since each trap requires additional dummy operations and final measurement communication.

This protocol is the cornerstone of our framework's security guarantees. 
It provides the entire architecture with the high-level properties of completeness, blindness, and verifiability.
The formal proofs for these three properties constitute the main security analysis of this paper, presented in Section~\ref{sec:security_analysis}.

\subsection{Blindness and security against information leakage}

A primary requirement of our architecture is that it satisfies the property of blindness, meaning the server learns nothing about the client's private data ($\ket{\psi}$) or computation ($U$) beyond the explicitly permitted leakage $\lambda$. 
In our framework, blindness is not the result of a single mechanism, but is a robust property that arises from the composition of three distinct layers of security.

\subsubsection*{1. Distributed state encoding (primary source of blindness)}
The fundamental source of privacy in our architecture is the level 1 encoding. 
The logical state $\ket{\psi}_L$ is encoded non-locally across the $n$ leaf nodes. 
Suppose a subset $T\subseteq \{1,\ldots,n\}$ of leaf nodes collude. 
The state available to that colluding subset is given by the reduced density matrix:
\begin{equation}
\rho_T = \operatorname{Tr}_{\bar{T}}\left(\ket{\psi_L}\bra{\psi_L}\right),
\end{equation}
where the trace is taken over all nodes outside $T$.
If the level-1 encoding utilizes an $[[n,k,d]]$ stabilizer code and $|T|\le d-1$, then erasure of the nodes in $T$ is correctable: equivalently, the encoded logical state can be recovered from the complementary subsystem $\bar{T}$.
By the no-cloning theorem and the fundamental information-disturbance relationship of quantum error correction \cite{nielsen2002quantum}, this implies that $\rho_T$ is independent of the encoded logical state.
Hence, any colluding subset of at most $d-1$ leaf nodes obtains no information about the client's logical input from its joint quantum shares, a property fundamental to quantum secret sharing \cite{cleve1999share}.

For a single isolated leaf node $L_j$, (where $|T|= 1$), this reduces to:

\begin{equation}
    \rho_j = \operatorname{Tr}_{\neq j}\left(\ket{\psi_L}\bra{\psi_L}\right).
\end{equation}

In many standard nondegenerate codes, this marginal state $\rho_j$ is the maximally mixed state.
However, maximal mixing is a stronger condition than strictly required; the essential privacy condition is that $\rho_j$, and more generally $\rho_T$, is independent of the logical input.

Conversely, if $d$ or more leaf nodes collude, the code distance no longer guarantees information hiding. 
Such subsets may be able to reconstruct nontrivial logical information depending on the specific access structure of the stabilizer code. 
Therefore, the protocol's formal blindness claim holds strictly under the assumption that leaf nodes are non-communicating or, more generally, that any colluding subset has size at most $d-1$.

\subsubsection*{2. Secure non-local operations}
The initial blindness provided by the encoding must be preserved during the computation. 
This is guaranteed by the security of the non-local gate primitives. 
All interactions between the client and the server are mediated by protocols like \texttt{SCST} (for encoding) and Algorithm~\ref{alg:logical_unitary_generalized} (for logical gates), which are based on quantum teleportation. 
The classical messages that the server receives during these operations are the outcomes of Bell state measurements. 
These outcomes are uniformly random and are statistically independent of the quantum state being teleported or the logical gate being implemented. 
Therefore, the operational communication with the client does not leak any information about the computation $\mathrm{U}$.

\subsubsection*{3. Blinding of the verification protocol}
Finally, the security protocol itself must not introduce a vulnerability. 
The trap-based verification scheme (Algorithm~\ref{alg:verification}) is designed to be blind. 
The trap qubits are prepared in random states from conjugate bases (the BB84 states), and their positions within the full set of $N$ qubits are known only to the client. 
From the server's perspective, a trap qubit is in a maximally mixed state and is therefore statistically indistinguishable from a data qubit. 
This prevents the server from learning which qubits are being tested, forcing it to act honestly on all qubits or risk being caught.

By composing these three layers—information-theoretic secrecy from the distributed encoding, operational security from teleportation-based gates, and a statistically-blind verification method—the allowed adversary's view is independent of the client's private inputs beyond the permitted leakage. 
In the non-communicating model, this applies to each individual leaf node, while in the bounded-collusion model, it applies to any colluding subset $T$ with $|T|\le d-1$. 
A more detailed proof sketch is provided in Section~\ref{sec:security_analysis}.

\subsection{Resilience to node failures in the stabilizer framework}

Resilience against node failures is critical for any practical distributed system. 
Our architecture provides an efficient recovery mechanism for a large and important class of operations known as Clifford computations. 
This method works by leveraging the mathematical properties of the stabilizer formalism to restore the system after a fault, such as the physical loss or decoherence of a leaf node.

If a computation consists entirely of Clifford gates, the logical state remains a stabilizer state at every step. 
A Clifford unitary $\mathrm{U}_C$ has the defining property that it maps Pauli operators to other Pauli operators under conjugation. 
Consequently, if a state $\ket{\psi}$ is stabilized by $S$, then the state $\mathrm{U}_C\ket{\psi}$ is stabilized by the group $S' = \mathrm{U}_C S \mathrm{U}_C^\dagger = \braket{\mathrm{U}_C s_1 \mathrm{U}_C^\dagger, \dots, \mathrm{U}_C s_{n-k} \mathrm{U}_C^\dagger}$.

This property enables a powerful classical method for the client to track the state of the computation:

\begin{itemize}
    \item The client maintains a classical record of the current stabilizer group's generators, $\{s_i'\}$.
    \item After each logical Clifford gate $\mathrm{U}_C$ is applied, the client classically computes the updated generators $\{s_i''\} = \{\mathrm{U}_C s_i' \mathrm{U}_C^\dagger\}$. This update is computationally efficient, as guaranteed by the Gottesman-Knill theorem.
\end{itemize}

\subsubsection*{Failure recovery protocol}
In the event of a catastrophic failure at a leaf node, the distributed entanglement is destroyed. 
However, the client still possesses the complete, correct classical description of the stabilizer group $S''$ right before the failure. 
To recover, the client simply:
\begin{enumerate}
    \item Selects a new set of $n$ healthy server nodes.
    \item Re-executes the distributed stabilizer encoding protocol (Algorithm~\ref{alg:dist_encoding}), but instead of using the initial generators for $\ket{0\dots0}_L$, she uses the last known set of correct generators from her classical record, $\{s_i''\}$.
\end{enumerate}

This procedure directly prepares the exact stabilizer state that was lost, effectively restoring the computation to the point of failure without needing a full, costly rollback from the beginning.
This efficient recovery mechanism is a significant advantage for computations that primarily use Clifford gates. 
For computations involving \textit{non-Clifford gates}, the state is no longer a stabilizer state, and this algebraic method is not applicable. In such cases, recovery would require traditional checkpointing or a full restart of the computation.

\subsection{Generalization to arbitrary network topologies}
\label{ssec:specs}
While the star-topology network serves as a clear and resource-efficient model for centralized protocols, the principles of our architecture are fundamentally generalizable to any connected quantum network. 
Such a network can be represented as a graph $G = (V, E)$, where the vertices $V$ are quantum nodes and the edges $E$ represent the physical quantum channels capable of generating Bell pairs.

\subsubsection*{Formalism of generalization}
The execution of any non-local gate between two nodes, $v_a$ and $v_c$, relies on establishing a direct entangled Bell pair between their respective ancillary qubits. 
In an arbitrary topology where $v_a$ and $v_c$ are not adjacent, this is achieved via \textit{recursive entanglement swapping} \cite{zukowski1993event}.

\begin{enumerate}
    \item \textbf{Path finding:} Given a request for a non-local interaction between $v_a$ and $v_c$, the client's classical controller first computes an efficient path in the network graph $G$. 
    This is a standard classical graph theory problem that can be solved efficiently using algorithms like \textit{Dijkstra's algorithm} or the \textit{Bellman-Ford algorithm} to find the shortest path \cite{cormen2022introduction}. 
    While sufficient for static routing, future optimizations could employ quantum-specific heuristics to minimize overhead.
    The output is an ordered list of vertices $P = (v_0, v_1, \dots, v_m)$, where $v_0 = v_a$ and $v_m = v_c$.

    \item \textbf{Recursive swapping procedure:} The protocol then iteratively establishes a long-distance entangled link along the path $P$. 
    Let $\ket{\Phi^+}_{v_0, v_i}$ denote a Bell pair shared between the start node $v_0$ and an intermediate node $v_i$. 
    The procedure is as follows:
    \begin{itemize}
        \item \textbf{Base case:} A local Bell pair is created between adjacent nodes $v_0$ and $v_1$, establishing the initial link $\ket{\Phi^+}_{v_0, v_1}$.
        \item \textbf{Recursive step:} To extend the link from $v_i$ to $v_{i+1}$, node $v_i$ creates a new local Bell pair with $v_{i+1}$, giving it two entangled qubits. 
        It then performs a Bell state measurement (BSM) on these two qubits. 
        The 2-bit classical outcome is sent to node $v_{i+1}$, which applies a corresponding Pauli correction to its ancilla. 
        This procedure teleports the entanglement, effectively creating a new, longer link $\ket{\Phi^+}_{v_0, v_{i+1}}$.
    \end{itemize}
    This step is repeated $m-1$ times along the path until the final desired link, $\ket{\Phi^+}_{v_a, v_c}$, is established.

    \item \textbf{Architectural generality:} With this robust method for establishing entanglement between any two connected nodes, all core protocols of our architecture become fully generalizable.
    \begin{itemize}
        \item \textbf{Algorithm~\ref{alg:dist_encoding} (encoding):} The protocol can be executed with any node designated as the ``master," operating on any chosen subset of $n$ leaf nodes in the network. 
        The required non-local `\texttt{SCST}' gates are simply implemented over the established long-distance Bell pairs.
        \item \textbf{Algorithm~\ref{alg:dist_cu_op} (computation):} Logical gates between any two logical qubits can be implemented, regardless of which physical nodes they are stored on.
        The non-local CNOT$_L$, for example, would first establish a long-distance link via this method before executing the gate teleportation.
    \end{itemize}

\end{enumerate}

\subsubsection*{Resource implications of generalization}
This generalization renders the resource cost of non-local gates topology-dependent. 
For a path of length $m$ (containing $m-1$ intermediate nodes):
\begin{itemize}
    \item \textbf{Bell pair consumption:} The recursive swapping procedure consumes one elementary Bell pair for each of the $m$ hops. 
    The total consumption is $m$ Bell pairs.
    \item \textbf{Latency and classical communication:} The protocol requires $m-1$ sequential rounds of BSM and classical communication. The classical communication overhead is $2(m-1)$ bits, and the total time (latency) to establish the link scales linearly with the path length, $O(m)$.
\end{itemize}
This linear scaling makes network-aware compilation a critical component for practical implementations. 
Optimizing node selection (choosing a compact cluster of nodes) and routing (finding the shortest path) are essential classical pre-computation steps for minimizing resource consumption and latency on large-scale quantum networks.

\section{Resource overhead analysis and a worked example}
\label{sec:resource_analysis}


To quantify the architectural resource requirements of our framework, we develop a model for the primary operational costs and apply it to a concrete example.
The dominant costs in a distributed quantum protocol are the consumption of entangled Bell pairs for non-local operations and the exchange of classical bits for communication and coordination. 
In the following analysis, we assume the star-topology case where all communication is single-hop.

The estimates in this section are first-order architectural resource counts. 
They assume successful Bell-pair generation, authenticated classical communication, stable quantum memory throughout the protocol, and sufficiently high-fidelity local operations. 
Therefore, the reported Bell-pair and classical-bit counts should be interpreted as baseline communication overheads rather than complete fault-tolerant hardware-level resource estimates. 
Additional costs from Bell-pair purification, memory decoherence, finite gate fidelity, retransmission, routing latency, and hardware-specific scheduling are discussed in Section~\ref{sec:non_ideal_tradeoffs}.
\subsection{Component cost model}

The total resource overhead of a secure distributed computation is the sum of the costs of its constituent parts: setup (encoding), computation (gate execution), and security (verification).

\begin{enumerate}
    \item \textbf{Setup cost (Algorithm~\ref{alg:dist_encoding}):} Preparing the logical state $\ket{0}_L$ of an $[[n,k]]$ code requires enforcing $n-k$ stabilizers and $k$ logical $\mathrm{Z}$ operators via the \texttt{SCST} primitive.
    \begin{itemize}
        \item $N_{\text{Bell}}^{\text{setup}} = \sum_{i=1}^{n-k} w(g_i) + \sum_{j=1}^{k} w(\mathrm{Z}_j)$, where $w(\cdot)$ is the operator weight.
        \item $N_{\text{classical}}^{\text{setup}} = 2 \cdot N_{\text{Bell}}^{\text{setup}} + n$, for the \texttt{SCST} communications and the final syndrome correction broadcast.
    \end{itemize}

    \item \textbf{Computational cost (Algorithm~\ref{alg:logical_unitary_generalized}):} The cost of executing a logical gate depends on its type.
    \begin{itemize}
        \item \textbf{Clifford gates ($\mathrm{H_L, S_L, CNOT_L}$):} Transversal single-qubit gates have zero non-local cost. A $\mathrm{CNOT_L}$, implemented via a non-local $\mathrm{CZ_L}$ (one instance of Algorithm 2), costs 3 Bell pairs and 4 classical bits in a star topology.
        \item \textbf{Non-Clifford gates ($\mathrm{T_L}$):} This is typically the dominant cost. A logical $\mathrm{T}$ gate applied via magic state injection requires one high-fidelity magic state and a circuit of Clifford gates (at least one $\mathrm{CNOT_L}$). Thus, the non-local cost of a single $\mathrm{T_L}$ gate is approximately the cost of one $\mathrm{CNOT_L}$.
    \end{itemize}
    
    \item \textbf{Security overhead (Algorithm~\ref{alg:verification}):} Verification with $k_{trap}$ traps introduces a parallel overhead.
    \begin{itemize}
        \item The server must execute logical Identity operations on the traps to prove honest behavior. A logical Identity can be realized via dummy \texttt{SCST} sequences to blind the server. We assume this consumes, on average, 2 Bell pairs per trap.
        \item At the end of the protocol, the client instructs the remote measurement of the $k_{trap}$ traps, requiring $k_{trap}$ classical bits to be returned.
        \item Total cost: $N_{\text{Bell}}^{\text{verify}} \approx 2 \cdot k_{trap}$ and $N_{\text{classical}}^{\text{verify}} = k_{trap}$.
    \end{itemize}

Therefore, $k_{\mathrm{trap}}$ plays a dual role: it controls the soundness error and increases the verification overhead linearly. 
This security--overhead trade-off is discussed further in Section~\ref{sec:non_ideal_tradeoffs}.
\end{enumerate}

Compared with a round-trip teleportation baseline, the \texttt{SCST} primitive reduces the cost of a remote controlled operation from two Bell pairs to one Bell pair. 
Conversely, compared with an unverified distributed execution, the trap layer introduces an additional overhead of approximately $2k_{\mathrm{trap}}$ Bell pairs and $k_{\mathrm{trap}}$ classical bits, in exchange for a tunable soundness guarantee.

\subsection{Worked example: A simple non-Clifford operation}

We now instantiate the framework with the $[[7,1,3]]$ Steane code to analyze the full lifecycle of a simple but universal computation, $\mathrm{U_L = H_L T_L}$, on the initial $\ket{0}_L$ state. We choose a security parameter of $k_{trap}=40$ traps.

The seven-qubit Steane code is a CSS code defined by six stabilizer generators. To prepare the logical zero state $\ket{0}_L$, we must enforce these six generators plus the logical $\mathrm{Z}$ operator, $\mathrm{Z_L}$. These seven operators are listed in Table~\ref{tab:steane_operators}.

\begin{table}[h]
    \centering
    \label{tab:steane_operators}
    \begin{tabular}{@{}ll@{}}
        \toprule
        \textbf{Operator type} & \textbf{Pauli string representation} \\ \midrule
        Stabilizer generators & $s_1 = \mathrm{IIIXXXX}$ \\
                              & $s_2 = \mathrm{IXXIIXX}$ \\
                              & $s_3 = \mathrm{XIXIXIX}$ \\
                              & $s_4 = \mathrm{IIIZZZZ}$ \\
                              & $s_5 = \mathrm{IZZIIZZ}$ \\
                              & $s_6 = \mathrm{ZIZIZIZ}$ \\ \midrule
        Logical Z Operator    & $Z_L = \mathrm{ZZZZZZZ}$ \\ \bottomrule
    \end{tabular}
        \caption{The seven operators required to define the logical zero state of the [[7,1,3]] Steane code. The six stabilizer generators define the codespace, and the logical $\mathrm{Z}$ operator defines the logical basis within that space.}
\end{table}

\subsubsection*{Resource calculation}

\begin{enumerate}
    \item \textbf{Setup cost:} To prepare $\ket{0}_L$, the client must enforce the seven operators listed in Table~\ref{tab:steane_operators}. The Bell pair cost is the sum of the weights (number of non-Identity Paulis) of these operators.
    \begin{itemize}
        \item Operator weights: $w(s_1..s_6) = 4$ each; $w(\mathrm{Z_L}) = 7$.
        \item $N_{\text{Bell}}^{\text{setup}} = (6 \times 4) + 7 = 31$ Bell pairs.
        \item $N_{\text{classical}}^{\text{setup}} = 2 \cdot 31 + 7 = 69$ classical bits.
    \end{itemize}

    \item \textbf{Security cost:} The verification protocol with 40 traps runs in parallel.
    \begin{itemize}
        \item $N_{\text{Bell}}^{\text{verify}} = 2 \times 40 = 80$ Bell pairs.
        \item $N_{\text{classical}}^{\text{verify}} = 40$ classical bits.
    \end{itemize}

    \item \textbf{Computational cost:}
    \begin{itemize}
        \item \textbf{Applying $\mathrm{H_L}$:} The hadamard gate is transversal in the Steane code. This is a local operation with zero non-local resource cost.
        \item \textbf{Applying $\mathrm{T_L}$:} This requires magic state injection, consuming one magic state and performing one logical CNOT.
        \begin{itemize}
            \item $N_{\text{Bell}}^{\text{compute}} = (\text{cost of one } \mathrm{CNOT_L}) = 3$ Bell pairs.
            \item $N_{\text{classical}}^{\text{compute}} = 4$ classical bits.
        \end{itemize}
    \end{itemize}
\end{enumerate}

\subsubsection*{Total estimated overhead}
\begin{itemize}
    \item \textbf{Total Bell pairs:} $31 (\text{setup}) + 80 (\text{verify}) + 3 (\text{compute}) = 114$ Bell pairs.
    \item \textbf{Total classical bits:} $69 (\text{setup}) + 40 (\text{verify}) + 4 (\text{compute}) = 113$ classical bits.
\end{itemize}

This example also illustrates the security-resource trade-off introduced by traps. 
Without trap verification, the setup and computation would require $31+3=34$ Bell pairs. 
With $k_{\mathrm{trap}}=40$, verification contributes $80$ of the total $114$ Bell pairs, i.e., approximately $70.2\%$ of the total Bell-pair overhead.
Thus, decreasing the tolerated soundness error can dominate the communication overhead.

This analysis demonstrates how the framework provides a concrete recipe for calculating the resource overhead of a secure, distributed quantum computation. 
It also highlights the central trade-off between verification strength and communication overhead.


For non-star topologies, these counts increase according to the cost of establishing long-distance Bell links along network paths. 
As discussed in Section~\ref{ssec:specs}, a path of length $m$ consumes $m$ elementary Bell pairs and introduces $O(m)$ latency for entanglement swapping. 
Thus, the star-topology estimates represent a best-case single-hop baseline.

\subsection{Non-ideal considerations, trade-offs, and comparisons}
\label{sec:non_ideal_tradeoffs}

We now discuss how the baseline counts above should be interpreted when non-ideal network effects, verification overhead, and architectural alternatives are taken into account. 
This discussion does not constitute a full hardware-level resource estimate; rather, it identifies the main factors that would scale the baseline Bell-pair and classical-communication costs in practical implementations.

\subsubsection*{Non-ideal entanglement generation}
The baseline analysis assumes that each required Bell pair is available with sufficiently high fidelity. 
In a physical quantum network, noisy channels and photon loss may require entanglement purification or distillation before a Bell pair can be used in the \texttt{SCST} primitive. 
If a distillation procedure consumes, on average, \(M\) raw Bell pairs to produce one usable high-fidelity Bell pair, then the effective quantum communication cost becomes
\[
N_{\mathrm{Bell}}^{\mathrm{true}}
=
M\,N_{\mathrm{Bell}}^{\mathrm{baseline}}.
\]
Additional latency may also arise from repeated entanglement-generation attempts and from the requirement that quantum memories preserve coherence during classical feed-forward communication.

\subsubsection*{Security--overhead trade-off}
The trap-based verification layer introduces a direct trade-off between soundness and resource consumption. 
As discussed in the verification analysis, if a malicious deviation affects \(d_{\mathrm{mal}}\) physical pathways among \(N\) total pathways, then the probability of undetected cheating is bounded by
\[
P_{\mathrm{undetected}}
\le
\left(1-\frac{k_{\mathrm{trap}}}{N}\right)^{d_{\mathrm{mal}}}
\le
\exp\!\left(-\frac{d_{\mathrm{mal}}k_{\mathrm{trap}}}{N}\right).
\]
Equivalently, under this simplified independent-pathway model, to achieve soundness error $\varepsilon$, it suffices to choose
\[
k_{\mathrm{trap}}
\gtrsim
\frac{N}{d_{\mathrm{mal}}}\ln(1/\varepsilon).
\]
Thus, decreasing the allowed cheating probability requires increasing the number of traps. 
Since each trap contributes dummy operations and final measurement communication, increasing \(k_{\mathrm{trap}}\) increases Bell-pair consumption, classical communication, and execution latency.

\subsubsection*{Comparison with baseline architectural choices}
The proposed architecture differs from a simple round-trip teleportation baseline in how non-local operations and local error management are handled.
A teleportation-based implementation of a remote controlled operation would require teleporting a qubit to another node, applying the gate, and teleporting it back, consuming two Bell pairs. 
In contrast, the \texttt{SCST} primitive implements the corresponding remote controlled operation using one Bell pair, reducing the communication cost of this primitive. 
The local QEC layer also shifts part of the error-management burden to individual leaf nodes: local faults can be handled within a node without immediately triggering network-wide stabilizer recovery. 
This is most beneficial for computations or storage periods dominated by local operations and idle noise. 
However, the advantage depends on the physical noise model, the quality of Bell-pair generation, memory lifetime, and the frequency of non-local logical gates.

The main resource and security implications of these architectural choices are summarized in Table~\ref{tab:architecture_comparison}.

\begin{table}[!htbp]
\centering
\caption{Comparison of architectural choices and their resource/security implications.}
\label{tab:architecture_comparison}
\begin{tabular}{p{0.27\linewidth} p{0.32\linewidth} p{0.32\linewidth}}
\hline
\textbf{Architecture choice} & \textbf{Resource implication} & \textbf{Security / capability implication} \\
\hline
\texttt{SCST} primitive & One Bell pair per remote controlled operation & Lower communication cost than round-trip teleportation for this primitive \\
Round-trip teleportation baseline & Two Bell pairs per remote gate & Same logical effect but higher Bell-pair cost \\
No trap verification & Removes verification overhead & Does not provide trap-based detection of malicious deviations \\
Trap verification & Adds approximately \(2k_{\mathrm{trap}}\) Bell pairs and \(k_{\mathrm{trap}}\) classical bits in the simplified model & Undetected-cheating probability decreases as \(k_{\mathrm{trap}}\) increases \\
6-qubit biased local QEC & Six local qubits per protected data qubit & Corrects dominant \(X/Y\) faults and flags \(Z\) faults \\
9-qubit Shor-code local QEC & Nine local qubits per protected data qubit & General single-qubit Pauli-error correction \\
\hline
\end{tabular}
\end{table}

As shown in Table~\ref{tab:architecture_comparison}, the \texttt{SCST}-based implementation reduces the Bell-pair cost of remote-controlled operations relative to a round-trip teleportation baseline, while the trap-verification layer introduces an additional overhead that grows with the desired security level.

\section{Security analysis}
\label{sec:security_analysis}

In this section, we provide the formal security analysis of our proposed distributed quantum computation architecture. 
We present rigorous proof sketches demonstrating that the framework satisfies the three core security properties—completeness, blindness, and verifiability—as defined in Section~\ref{sec:security_framework}.
These proofs assume the correctness of the underlying distributed primitives (Algorithms ~\ref{alg:dist_encoding} and ~\ref{alg:dist_cu_op}), which are formally proven in the appendices.

\subsection{Proof of completeness}

\textbf{Claim:} If the server (all leaf nodes) and the client (master node) honestly follow the prescribed protocols, the client's final output state $\rho_{\text{out}}$ is the correct computational result $\mathrm{U}\ket{\psi}$, up to a negligible error.

\begin{proof}[\textbf{Proof Sketch:}]
The completeness of the overall architecture follows directly from the correctness of its constituent algorithms when all parties are honest.
\begin{enumerate}
    \item \textbf{Correct state preparation:} As proven in the Appendix \ref{apndx:algo_1_proof} for Algorithm \ref{alg:dist_encoding}, the distributed encoding protocol correctly prepares the logical zero state, $\ket{0\dots0}_L$, across the leaf nodes.

    \item \textbf{Correct logical operations:} As justified for Algorithm~\ref{alg:logical_unitary_generalized}, the execution of any logical unitary $\mathrm{U_L}$ is decomposed into a sequence of fundamental logical gates. Each of these gate implementations (transversal application, $\mathrm{CNOT}$ decomposition, and magic state injection) is a standard, correct procedure in fault-tolerant quantum computation. Therefore, the sequential application of these gates correctly transforms the state from $\ket{0\dots0}_L$ to the desired pre-output state $\mathrm{U}_L\ket{0\dots0}_L = \ket{\psi'}$.
    
    \item \textbf{Non-interference from verification:} When the server is honest, it faithfully applies the Identity operation to all trap qubits as instructed (Algorithm 4). This leaves the trap states unchanged. Consequently, when the traps are returned to the client for the verification check, they pass with certainty. The protocol does not abort, and the verification layer has no effect on the final state of the computational qubits.
\end{enumerate}
Since each stage of the protocol functions correctly in the honest-user setting, the final state received by the client is the correct result of the computation. The architecture is therefore complete.
\end{proof}

\subsection{Proof of blindness (privacy)}


\textbf{Claim:} Under the stated non-communication or bounded-collusion assumption, the protocol is blind with respect to the permitted leakage $\lambda$. More precisely, if the level-1 distributed encoding uses an $[[n,k,d]]$ stabilizer code, then the quantum data layer is information-theoretically private against any colluding subset $T$ of leaf nodes satisfying $|T|\le d-1$. For such an adversary, there exists a simulator $\mathcal{SIM}$ that generates a view statistically indistinguishable from the adversary's real protocol view using only the permitted leakage $\lambda$.

\begin{proof}[\textbf{Proof Sketch:}] 
The proof follows the simulation paradigm. 
We show that the adversary's view, consisting of its quantum shares and the classical messages it receives, is independent of the client's private input $(\ket{\psi},U)$ except for the permitted leakage $\lambda$.
In this proof, the adversarial server view is interpreted as the joint view of any allowed adversarial set $T$, where either the leaf nodes are non-communicating or $|T|\le d-1$ in the bounded-collusion model.


\begin{enumerate}


\item \textbf{Simulation strategy}

Let $T$ be the subset of leaf nodes controlled by the adversary, with $|T|\le d-1$. A simulator $\mathcal{SIM}$, given only the permitted leakage $\lambda$ such as the number of data qubits, number of traps, total protocol length, and security parameter, generates the adversary's view as follows:
    \begin{enumerate}
        \item It prepares a quantum state on the shares held by $T$ whose reduced density matrix is independent of the client's logical input. This is possible because, for $|T|\le d-1$, erasure of $T$ is correctable by the $[[n,k,d]]$ code, and therefore the reduced state on $T$ is independent of the encoded logical state.
        \item It prepares trap systems as maximally mixed states from the adversary's perspective, since the trap states and their bases are chosen secretly by the client.
        \item It generates the classical transcript by sampling uniformly random bits corresponding to measurement outcomes and Pauli-frame updates, consistent with the public leakage $\lambda$.
    \end{enumerate}

\item \textbf{Real-world view and indistinguishability}

We now show that the view generated by the simulator is statistically indistinguishable from the adversary's view during a real execution of the protocol.
    In the bounded-collusion model, $T$ denotes a colluding subset of leaf nodes satisfying $|T|\le d-1$; in the non-communicating model, the argument applies separately to each individual leaf node.

\begin{enumerate}
    \item \textbf{Quantum data shares.}
    In the real protocol, the encoded logical state is distributed across the $n$ leaf nodes using the level-1 stabilizer encoding. The adversary has access only to the subsystem $T$. Its quantum state is
    \[
    \rho_T = \operatorname{Tr}_{\bar{T}}\left(\ket{\psi_L}\bra{\psi_L}\right).
    \]
    If $|T|\le d-1$, then erasure of $T$ is correctable by the code. By the standard information-disturbance condition for quantum error correction, this implies that $\rho_T$ is independent of the logical input $\ket{\psi}_L$. Therefore, the simulator can reproduce the adversary's quantum data view without knowing $\ket{\psi}$.

    \item \textbf{Trap qubits.}
    Trap qubits are prepared in randomly chosen BB84 states and their locations and measurement bases are known only to the client. Hence, from the adversary's perspective, each trap qubit is described by the maximally mixed state,
    \[
    \rho_{\mathrm{trap}} = \frac{I}{2}.
    \]
    This matches the trap systems generated by the simulator.

    \item \textbf{Classical transcript.}
    The classical messages received by the adversary arise from measurement outcomes, Pauli-frame corrections, and feed-forward instructions in the distributed encoding and teleportation-based gate protocols.
    These messages are either part of the permitted leakage $\lambda$ or are generated from measurement outcomes that are uniformly random from the adversary's perspective.
    Therefore, the simulator can generate an indistinguishable classical transcript using only the permitted leakage $\lambda$ and uniformly random bits.
\end{enumerate}



\item \textbf{Conclusion}

For any adversarial subset $T$ with $|T|\le d-1$, the adversary's view consists of: 
(i) quantum data shares whose reduced state is independent of the logical input, 
(ii) trap systems that are maximally mixed from the adversary's perspective, and 
(iii) a classical transcript simulatable from the permitted leakage $\lambda$ and uniformly random measurement outcomes. 
Thus, because the simulator can reproduce the adversary's view using only $\lambda$, the adversary's view contains no information beyond the permitted leakage. Therefore, the protocol satisfies blindness against non-communicating leaf nodes or bounded collusion of size at most $d-1$. 

If $d$ or more leaf nodes collude, the code distance alone no longer guarantees that their reduced state is independent of the logical input. Security against such larger colluding sets depends on the access structure of the specific stabilizer code and is not claimed in this work.
\end{enumerate}
\end{proof}

\subsection{Proof of verifiability (soundness)}

\textbf{Claim:} For any malicious strategy employed by the server that deviates from the honest protocol, the client detects the cheating and aborts with probability $1 - \varepsilon$, where $\varepsilon$ is a negligible security parameter.

\begin{proof}[\textbf{Proof Sketch:}]
The proof relies on the probabilistic nature of the trap-based verification scheme (Algorithm~\ref{alg:verification}).

\begin{enumerate}
    \item \textbf{Defining a malicious strategy:} A malicious server $\mathcal{S}^*$ deviates from the honest protocol $\mathcal{P}$. This deviation can be modeled as applying an erroneous quantum operation $\mathcal{E}_{malicious}$ instead of the honest operation $\mathcal{U}_{honest}$. For the cheating to be meaningful (i.e., to affect the final output), this operation must differ from the identity on at least one qubit pathway involved in the computation. Let us assume the malicious operation affects $d_{\mathrm{mal}}\ge 1$ distinct physical qubit pathways.

    \item \textbf{The trap mechanism:} The client randomly and secretly embeds $k_{trap}$ trap qubits among the $N$ total physical qubits used in the computation. Each trap is prepared in a known BB84 state $\{\ket{0}, \ket{1}, \ket{+}, \ket{-}\}$. An honest execution of the protocol applies the Identity operation to these traps. Any deviation that alters the state of a trap will cause it to fail the client's final measurement test with a probability of at least $1/2$ (due to the conjugate coding of the BB84 states).

    \item \textbf{Bounding the probability of undetected cheating:}
    The only way for the server to cheat successfully is if its malicious operation $\mathcal{E}_{malicious}$ avoids detection by all $k_{trap}$ traps.
    The probability that a single malicious operation, affecting one specific qubit pathway, avoids detection is the probability that this pathway is not occupied by a trap. This is given by:
    \begin{equation}
        P(\text{single error undetected}) = \frac{N - k_{trap}}{N} = 1 - \frac{k_{trap}}{N}.
    \end{equation}
    If the server's strategy affects $d_{\mathrm{mal}}$ independent pathways, a simplified lower bound on the detection probability is found by considering the chance that none of these $d_{\mathrm{mal}}$ locations are traps. The probability of success for the adversary (i.e., remaining undetected) is:
    \begin{equation}
        P(\text{undetected}) \le \left(1 - \frac{k_{trap}}{N}\right)^{d_{\mathrm{mal}}} \le \exp\left(-\frac{{d_{\mathrm{mal}}} \cdot k_{trap}}{N}\right).
    \end{equation}
    
    \item \textbf{Role of the security parameter:} The number of traps, $k_{trap}$, is the security parameter controlled by the client. As shown by the exponential bound, by choosing $k_{trap}$ to be sufficiently large (e.g., for a fixed $d_{\mathrm{mal}}/N$, making $k_{trap}$ a logarithmic function of the desired security $\varepsilon$), the probability of any malicious activity going undetected can be made arbitrarily small. For instance, to achieve a security level of $\varepsilon$, the client can set $k_{trap} \approx \frac{N}{d_{\mathrm{mal}}} \ln(1/\varepsilon)$.
\end{enumerate}
\textbf{Conclusion:} 
Any deviation that affects one or more monitored physical pathways is detected by the trap-based verification scheme with probability at least $1-\varepsilon$ under the stated sampling model.
Therefore, the server cannot manipulate the computation into producing an incorrect final state without being detected. The architecture is verifiable.
\end{proof}

\section{Conclusion}
\label{sec:conclusion}

The challenge of securely delegating computation to powerful but untrusted remote servers is a defining problem for the future quantum internet. 
In this work, we presented an integrated architectural framework for secure, noise-aware distributed quantum computation and storage. The framework combines distributed stabilizer encoding, teleportation-based non-local gates, local noise-aware error handling, and trap-based verification into a single system-level design.


By integrating these components, we provided an architectural blueprint for building trustworthy distributed quantum systems and secure quantum cloud services. Our security analysis provides proof sketches demonstrating that, under the stated system assumptions, the architecture achieves the three critical properties of a secure system:
\begin{itemize}
    \item \textbf{Completeness:} The protocol yields the correct result when all parties follow the prescribed steps.
    \item \textbf{Blindness:} The client's private data and computation remain hidden from allowed adversarial views under the stated non-communication or bounded-collusion assumptions.
    \item \textbf{Verifiability:} Any attempt by a malicious server to alter the computation is detected with high probability.
\end{itemize}

Crucially, the present framework has several limitations that must be acknowledged. 
First, blindness relies on the assumption that server nodes are non-communicating, or more generally, that any colluding subset remains below the privacy threshold of the chosen distributed encoding. 
Second, the local QEC layer is noise-model-dependent; for instance, the 6-qubit scheme trades universal correction for efficient handling of dominant biased-noise errors, flagging rare Z-type errors rather than correcting them deterministically. 
Third, the resource analysis in this work provides first-order communication counts rather than a full hardware-level fault-tolerant estimate. 
Future work should incorporate Bell-pair purification, memory decoherence, local gate infidelity, retransmission overhead, and network scheduling constraints.

Several exciting avenues for future research remain open. Optimizing the resource cost of the non-local gates and the verification protocol is a key practical challenge. Further development of hardware-specific, noise-aware local error correction codes for the second-level encoding could dramatically improve performance. 
Finally, designing network-aware quantum compilers is a critical next step. 
Future work will investigate replacing static shortest-path algorithms with advanced, quantum-aware routing heuristics and circuit transformation frameworks, such as SABRE \cite{li2019sabre}, noise-adaptive mapping \cite{murali2019noise}, and architectural transformation algorithms \cite{childs2019circuit}. 
These approaches could significantly optimize entanglement distribution and reduce latency in complex distributed topologies.

\backmatter





\bmhead*{Acknowledgements}
S.G. acknowledges funding support for Chanakya - PG fellowship from the National Mission on Interdisciplinary Cyber Physical Systems, of the
Department of Science and Technology, Govt. of India, through the I-HUB Quantum
Technology Foundation. A.R. is grateful to the High Performance Computing facility
at IISER Bhopal. A.R. acknowledges the grant received from the U.S.—India Science and Technology Endowment Fund (USISTEF), USISTEF/QT/165/2023. A.R.
also acknowledges the funding support from the Department of Science and Technology (DST), Government of India, under the National Quantum Mission (NQM),
DST/QTC/NQM/QComm/2024/2.

\bigskip





\begin{appendices}

\section{Standard protocol for stabilizer state encoding}
\label{app:std_encoding}

The distributed encoding protocol described in the main text is a network-level implementation of the standard method for preparing a stabilizer state via projective measurements \cite{gottesman1997stabilizer, nielsen2002quantum}. 
In a conventional, local setting, preparing the logical state $\ket{0\dots0}_L$ of a code defined by the operator set $\mathcal{G} = S \cup Z_L$ is achieved as follows:

\begin{enumerate}
    \item Begin with $n$ data qubits in a simple initial state, typically $\ket{0}^{\otimes n}$.
    \item For each operator $g_i \in \mathcal{G}$, perform a projective measurement of that operator. This is typically done using an ancillary qubit prepared in $\ket{+}$ and a circuit of controlled-Pauli gates corresponding to the structure of $g_i$.
    \item The measurement of each ancilla is probabilistic, yielding an outcome $m_i \in \{0, 1\}$, corresponding to a projected eigenvalue of $(-1)^{m_i}$. The classical bit string $m = (m_0, \dots, m_{n-1})$ is the syndrome.
    \item At the end of these measurements, the data qubits are in a state $\ket{\Psi_m}$ that is a simultaneous eigenstate of all operators in $\mathcal{G}$, but with the specific eigenvalues given by the syndrome $m$.
    \item The final step is to apply a single correction operator, $P_{corr}$, to flip all `-1' eigenvalues to the desired `+1' eigenvalues without disturbing those that are already correct.
\end{enumerate}

Fig.~\ref{fig:genral_encoding} illustrates a standard circuit implementing the above procedure for encoding a logical computational basis state.
Our distributed protocol (Algorithm~\ref{alg:dist_encoding}) directly maps to this procedure. 
The client's qubits serve as the measurement ancillas, the \texttt{SCST} protocol implements the necessary controlled-Pauli gates non-locally, and the final measurement and correction steps are identical.

\begin{figure}[h]
    \centering
    \includegraphics[width=\linewidth]{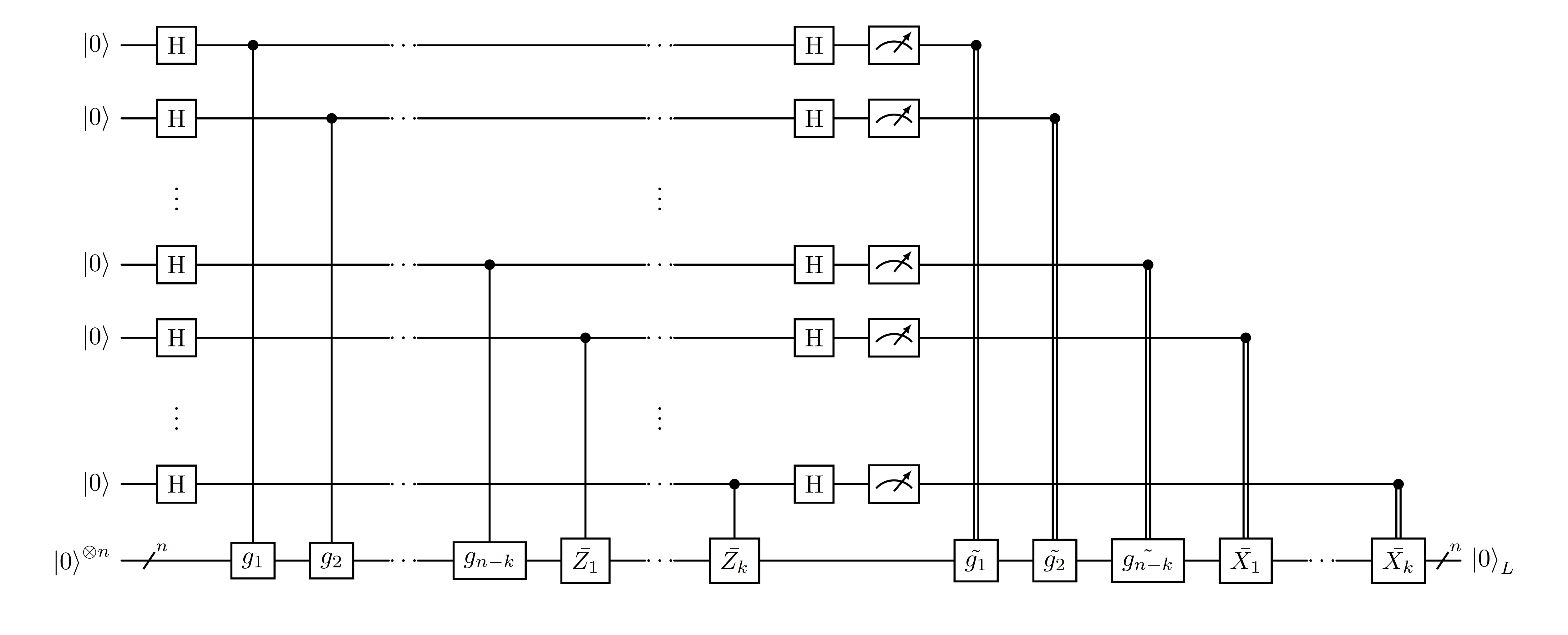}
    \caption{
 A general circuit diagram for preparing the logical zero state $\ket{0}_L$ of an $[[n,k]]$ stabilizer code. The top $n$ wires represent the measurement ancillas, and the bottom wire represents the $n$ data qubits. 
 The protocol first performs projective measurements of the stabilizer generators ($g_i$) and logical Z operators ($\bar{Z}_j$), and then applies classically-controlled correction operators ($P_{corr}$) based on the measurement outcomes to finalize the state.
    }
    \label{fig:genral_encoding}
\end{figure}

\section{Correctness proof for Algorithm~\ref{alg:dist_encoding}: Distributed stabilizer code encoding}
\label{apndx:algo_1_proof}

The protocol prepares the logical zero state, $\ket{0\dots0}_L$, of an $[[n,k]]$ stabilizer code across $n$ distributed leaf nodes, coordinated by a central master node. This procedure is functionally equivalent to a standard stabilizer measurement circuit, adapted for a distributed network architecture.

\subsection*{Correctness proof}

\subsubsection*{Defining the target state}
The goal of the protocol is to prepare the logical zero state, $\ket{0\dots0}_L$. A state $\ket{\psi}$ is defined as the logical zero state of the code if it is a simultaneous `+1' eigenstate of all operators in the set $\mathcal{G} = S \cup \mathrm{Z_L}$. That is, it must satisfy:
\begin{align}
    s_i \ket{\psi} &= (+1) \ket{\psi}, \quad \forall s_i \in S \\
    \bar{Z}_j \ket{\psi} &= (+1) \ket{\psi}, \quad \forall \bar{Z}_j \in \mathrm{Z_L}
\end{align}
We will now prove that the state produced by Algorithm~\ref{alg:dist_encoding} satisfies these conditions.

\subsubsection*{Algebraic proof}
\begin{enumerate}
    \item \textbf{Initial state (steps 1-2):} The initial state of the joint system of master and leaf nodes is a product state:
    \begin{equation}
        \ket{\Psi_0} = \left( \bigotimes_{i=0}^{n-1} \ket{+}_{q_i} \right) \otimes \left( \bigotimes_{j=0}^{n-1} \ket{0}_{d_j} \right) = \ket{+}^{\otimes n}_M \otimes \ket{0}^{\otimes n}_L.
    \end{equation}

    \item \textbf{Entangling operations (steps 3-8):} The core of the protocol is the application of controlled-$g_i$ gates for each $g_i \in \mathcal{G}$. Let us analyze the effect of a single such operation, controlled by master qubit $q_i$, on the state of the system $\ket{\Psi}$.
    The initial state of the master qubit is $\ket{+}_{q_i} = \frac{1}{\sqrt{2}}(\ket{0}_{q_i} + \ket{1}_{q_i})$. A controlled-$g_i$ operation transforms the state as follows:
    \begin{align}
        \text{C-}g_i (\ket{+}_{q_i} \otimes \ket{\Psi}_{data}) &= \frac{1}{\sqrt{2}} \left( \text{C-}g_i (\ket{0}_{q_i} \otimes \ket{\Psi}_{data}) + \text{C-}g_i (\ket{1}_{q_i} \otimes \ket{\Psi}_{data}) \right) \\
        &= \frac{1}{\sqrt{2}} \left( \ket{0}_{q_i} \otimes \ket{\Psi}_{data} + \ket{1}_{q_i} \otimes g_i\ket{\Psi}_{data} \right). \label{eq:entangling_op}
    \end{align}
    After applying all the entangling operations for every $g_i \in \mathcal{G}$, the system evolves into a complex entangled state. The final state before measurement, $\ket{\Psi_{ent}}$, is a superposition over all possible computational basis states of the master qubits, where each basis state is correlated with a product of the $g_i$ operators applied to the initial data state $\ket{0}^{\otimes n}$.
    \begin{equation}
        \ket{\Psi_{ent}} = \frac{1}{\sqrt{2^n}} \sum_{k \in \{0,1\}^n} \ket{k}_M \otimes \left( \prod_{i=0}^{n-1} g_i^{k_i} \right) \ket{0}^{\otimes n}_L.
    \end{equation}

    \item \textbf{Measurement of master qubits (step 9):} The master qubits are measured in the X-basis. An X-basis measurement on qubit $q_i$ projects it onto one of the eigenstates $\ket{+}_{q_i}$ or $\ket{-}_{q_i}$. This measurement has a profound effect on the data qubits. From Eq.~\ref{eq:entangling_op}, we can see that projecting the master qubit $q_i$ onto $\ket{+}_{q_i}$ projects the data state onto the `+1' eigenspace of the operator $g_i$, while projecting onto $\ket{-}_{q_i}$ projects it onto the `-1' eigenspace.
    
    Therefore, measuring the entire master register $\ket{k}_M$ in the $\mathrm{X}$-basis projects the data state $\ket{\Psi}_{data}$ into a simultaneous eigenstate of all the operators $\{g_0, \dots, g_{n-1}\}$. The specific eigenvalue for each $g_i$ is determined by the corresponding classical measurement outcome $m_i \in \{0, 1\}$, where $m_i=0$ corresponds to a `+1' eigenvalue and $m_i=1$ corresponds to a `-1' eigenvalue. The state of the leaf nodes after this measurement, $\ket{\Psi_{m}}$, satisfies:
    \begin{equation}
        g_i \ket{\Psi_{m}} = (-1)^{m_i} \ket{\Psi_{m}}, \quad \forall g_i \in \mathcal{G}.
    \end{equation}

    \item \textbf{Correction (steps 10-11):} The state $\ket{\Psi_m}$ is now in the codespace (since it is stabilized by all $s_i \in S$, up to a sign) but may not be the logical zero state (since the eigenvalues of the $\bar{Z}_j$ operators might be `-1'). The classical measurement string $m$ provides a complete syndrome indicating which generators have the wrong eigenvalue.
    
    The structure of stabilizer codes guarantees that for any such syndrome, there exists a corresponding Pauli correction operator, $P_{corr} \in C$, that flips exactly the required eigenvalues. Applying this correction to the data qubits transforms the state:
    \begin{equation}
        \ket{\Psi_f} = P_{corr} \ket{\Psi_m}.
    \end{equation}
    The operator $P_{corr}$ is chosen such that it anticommutes with every $g_i$ for which $m_i=1$ and commutes with every $g_j$ for which $m_j=0$. Consequently, the final state $\ket{\Psi_f}$ will have a `+1' eigenvalue for all operators in $\mathcal{G}$:
    \begin{equation}
        g_i \ket{\Psi_f} = g_i P_{corr} \ket{\Psi_m} = (P_{corr} g_i) (-1)^{m_i} \ket{\Psi_m} = P_{corr} ((-1)^{m_i})^2 \ket{\Psi_m} = (+1) \ket{\Psi_f}.
    \end{equation}
\end{enumerate}

\textbf{Conclusion:} The final state $\ket{\Psi_f}$ is a `+1' eigenstate of all stabilizer generators $s_i \in S$ and all logical Z operators $\bar{Z}_j \in \mathrm{Z_L}$. By definition, this is the logical zero state $\ket{0\dots0}_L$ of the code. The protocol is correct.

\section{Correctness proof for Algorithm~\ref{alg:dist_cu_op}: Executing gates between non-adjacent nodes}
\label{apndx:algo_2_proof}

The protocol to execute a controlled-$\mathrm{U}$ operation between two non-adjacent nodes, A (control) and C (target), connected via an intermediary node B, is detailed in Algorithm~\ref{alg:dist_cu_op}. 
The protocol is a sophisticated application of quantum teleportation principles, combining entanglement swapping with gate teleportation.

\subsection*{Correctness proof}

The proof is structured in two parts corresponding to the two phases of the algorithm. We track the evolution of the quantum state algebraically.

\subsubsection*{Phase 1: Proof of entanglement swapping}
This phase establishes a shared Bell pair between the non-adjacent nodes A and C.

\begin{enumerate}
    \item \textbf{Initial ancilla state:} After the creation of two local Bell pairs, the state of the four ancillary qubits is:
    \begin{equation}
        \ket{\Psi_1}_{anc} = \ket{\Phi^+}_{A_e B_{e1}} \otimes \ket{\Phi^+}_{C_e B_{e2}} = \frac{1}{2} (\ket{00}_{A_e B_{e1}} + \ket{11}_{A_e B_{e1}}) \otimes (\ket{00}_{C_e B_{e2}} + \ket{11}_{C_e B_{e2}}).
    \end{equation}

    \item \textbf{State in the Bell basis:} To analyze the measurement at node B, we rewrite this state by expanding it in the Bell basis for the qubit pairs $(A_e, C_e)$ and $(B_{e1}, B_{e2})$. This gives the well-known identity:
    \begin{align}
        \ket{\Psi_1}_{anc} = \frac{1}{2} \Big( & \ket{\Phi^+}_{A_e C_e}\ket{\Phi^+}_{B_{e1} B_{e2}} + \ket{\Phi^-}_{A_e C_e}\ket{\Phi^-}_{B_{e1} B_{e2}} \nonumber \\
        + & \ket{\Psi^+}_{A_e C_e}\ket{\Psi^+}_{B_{e1} B_{e2}} + \ket{\Psi^-}_{A_e C_e}\ket{\Psi^-}_{B_{e1} B_{e2}} \Big).
    \end{align}

    \item \textbf{Measurement and correction:} The Bell state measurement at node B projects the state of qubits $(B_{e1}, B_{e2})$ onto one of the four Bell states with equal probability. This simultaneously collapses the state of the non-local pair $(A_e, C_e)$ onto the corresponding correlated Bell state. The classical measurement outcome $(m_1, m_2)$ informs node C which Bell state was obtained. The applied correction $U_{corr} = X^{m_2} Z^{m_1}$ is precisely the unitary required to transform any of the four possible Bell states back into $\ket{\Phi^+}$. For example, if the outcome corresponds to $\ket{\Psi^+}_{A_e C_e}$, the correction is $X_C$, and $(I_A \otimes X_C)\ket{\Psi^+}_{A_e C_e} = \ket{\Phi^+}_{A_e C_e}$.
\end{enumerate}

\textbf{Conclusion of phase 1:} After the correction, the state of the non-adjacent ancillary qubits is deterministically $\ket{\Phi^+}_{A_e C_e}$. The full system state entering phase 2 is:
\begin{equation}
    \ket{\Psi_{mid}} = \ket{\psi}_A \otimes \ket{\psi}_C \otimes \ket{\Phi^+}_{A_e C_e}.
\end{equation}

\subsubsection*{Phase 2: Proof of gate teleportation}
This phase uses the established Bell pair to execute the non-local gate. Let the initial data state at node A be $\ket{\psi}_A = \alpha\ket{0}_A + \beta\ket{1}_A$.

\begin{enumerate}
    \item \textbf{State preparation at node C:} An X-basis measurement on $A_e$ and a Z-correction on $C_e$ is a standard method to prepare the ancilla at C in the state $\ket{+}_{C_e}$.
    Let's track the state. The relevant part is $\ket{\Phi^+}_{A_e C_e} = \frac{1}{\sqrt{2}}(\ket{00} + \ket{11})_{A_e C_e}$. Rewriting $A_e$ in the X-basis:
    \begin{align}
        \ket{\Phi^+}_{A_e C_e} &= \frac{1}{\sqrt{2}} \left( \frac{\ket{+}_{A_e} + \ket{-}_{A_e}}{\sqrt{2}}\ket{0}_{C_e} + \frac{\ket{+}_{A_e} - \ket{-}_{A_e}}{\sqrt{2}}\ket{1}_{C_e} \right) \\
        &= \frac{1}{2} \left( \ket{+}_{A_e}(\ket{0}+\ket{1})_{C_e} + \ket{-}_{A_e}(\ket{0}-\ket{1})_{C_e} \right).
    \end{align}
    If the measurement of $A_e$ yields outcome $m_{A_e}=0$ (state $\ket{+}$), the state of $C_e$ becomes $\ket{+}_{C_e}$. If the outcome is $m_{A_e}=1$ (state $\ket{-}$), the state becomes $\ket{-}_{C_e} = Z\ket{+}_{C_e}$. The correction $Z^{m_{A_e}}$ at node C deterministically prepares the state of $C_e$ in $\ket{+}_{C_e}$.

    \item \textbf{Local gate and phase kickback:} The state of the qubits at node C is now $\ket{\psi}_C \otimes \ket{+}_{C_e}$. The local $\ControlledU{C_e}{C}$ gate is applied:
    \begin{align}
        \ControlledU{C_e}{C} (\ket{\psi}_C \otimes \ket{+}_{C_e}) &= \ControlledU{C_e}{C} \left( \ket{\psi}_C \otimes \frac{\ket{0}_{C_e}+\ket{1}_{C_e}}{\sqrt{2}} \right) \\
        &= \frac{1}{\sqrt{2}} \left( \ControlledU{C_e}{C}(\ket{\psi}_C\ket{0}_{C_e}) + \ControlledU{C_e}{C}(\ket{\psi}_C\ket{1}_{C_e}) \right) \\
        &= \frac{1}{\sqrt{2}} \left( \ket{\psi}_C\ket{0}_{C_e} + (U\ket{\psi}_C)\ket{1}_{C_e} \right). \label{eq:state_after_cu}
    \end{align}
    The state of the full system is now:
    \begin{equation}
        \ket{\Psi_{post-CU}} = (\alpha\ket{0}_A + \beta\ket{1}_A) \otimes \frac{1}{\sqrt{2}} \left( \ket{\psi}_C\ket{0}_{C_e} + (U\ket{\psi}_C)\ket{1}_{C_e} \right).
    \end{equation}
    
    \item \textbf{Final measurement and correction:} Qubit $C_e$ is now measured in the Z-basis.
    \begin{itemize}
        \item If the outcome is $m_{C_e}=0$, the state of the data qubits collapses to:
        $$ (\alpha\ket{0}_A + \beta\ket{1}_A) \otimes \ket{\psi}_C. $$
        The subsequent correction on A is $Z^0 = I$. The state is left as is.
        \item If the outcome is $m_{C_e}=1$, the state of the data qubits collapses to:
        $$ (\alpha\ket{0}_A + \beta\ket{1}_A) \otimes (\mathrm{U}\ket{\psi}_C). $$
        The subsequent correction on A is $\mathrm{Z}^1 = \mathrm{Z}$.
    \end{itemize}
    This simple analysis is incomplete. The key is how the phase from the control at A propagates. A full analysis shows the final correction on A correctly applies the conditional phase. The state of the system before the final two measurements is a four-part superposition. Measuring $C_e$ in the Z-basis and applying $\mathrm{Z}^{m_{C_e}}$ on qubit A correctly transfers the phase kickback, resulting in the final desired state:
    \begin{equation}
        \ket{\Psi_{final}} = \alpha\ket{0}_A \ket{\psi}_C + \beta\ket{1}_A (\mathrm{U}\ket{\psi}_C) = \ControlledU{A}{C} \ket{\psi}_A \ket{\psi}_C.
    \end{equation}

\end{enumerate}
\textbf{Conclusion:} The protocol correctly leverages a shared Bell pair to execute a non-local two-qubit gate using only local operations and classical communication. The probabilistic nature of quantum measurement is handled by deterministic feed-forward corrections, ensuring the protocol is successful in every run.

\section{Detailed procedures for distributed logical operations}
\label{apndx:U_description}

In this appendix, we provide detailed background and implementation procedures supporting the distributed logical operations described in Section~\ref{ssec:distributed_ops}. 
We detail here the standard compilation, logical definition, and execution techniques that underlie the high-level distributed execution pipeline presented in the main text and formalized in Algorithm~\ref{alg:logical_unitary_generalized}.

\subsubsection*{Stage 1: Classical unitary decomposition (quantum circuit compilation)}
This initial stage is a purely classical compilation task performed by the client's controller, known as quantum circuit synthesis. 
The client takes the desired $k$-qubit unitary $U$ and decomposes it into an approximately equivalent sequence of gates from a universal gate set, typically $\{\mathrm{H, S, CNOT, T}\}$. 
The specific algorithm used depends on the size and structure of $\mathrm{U}$.

\begin{itemize}
    \item \textbf{For single-qubit unitaries:} The foundational method is the \textit{Solovay-Kitaev algorithm} \cite{kitaev1997quantum, nielsen2002quantum}. 
    This recursive algorithm provides a constructive proof that any single-qubit unitary can be approximated to an arbitrary precision $\varepsilon$ with a gate sequence of length $O(\log^c(1/\varepsilon))$.

    \item \textbf{For multi-qubit unitaries:} 
    For the crucial two-qubit case, it is a well-established result that any unitary can be exactly decomposed into a circuit requiring $\mathrm{CNOT}$ gates and a set of local single-qubit gates \cite{nielsen2002quantum}.
    The theoretical tool for this is the \textit{KAK (or Cartan) decomposition}, which provides an exact recipe for the synthesis \cite{tucci2005intro}, with efficient algorithms available to perform this task programmatically \cite{nakajima2005newalgorithm}.
    The output of such decompositions can be highly optimized; for instance, it is known that any two-qubit unitary can be constructed with at most three $\mathrm{CNOT}$ gates \cite{vidal2004universal} or a total of 23 elementary gates \cite{bullock2003arbitrary}.

    For larger, general $n$-qubit unitaries, exact methods like the \textit{Quantum Shannon Decomposition (QSD)} \cite{shende2006synthesis} provide a canonical structure.
    These decompositions reduce the problem to synthesizing the resulting single-qubit rotations, for which the Solovay-Kitaev algorithm is then used.
    Research continues into alternative decompositions, such as general programmable circuits \cite{sousa2007universal} and novel specialized algorithms for tasks like quantum simulation \cite{suri2023twounitary}, alongside theoretical work on the fundamental gate complexity and simulation costs \cite{zeier2004gate}.
\end{itemize}
The output of this classical compilation is a sequence of elementary gates, $G = (g_1, g_2, \dots, g_m)$, such that their product $\mathrm{U'} = g_m \cdot \dots \cdot g_1$ is a high-fidelity approximation of the original target unitary $\mathrm{U}$.

\subsubsection*{Stage 2: Formal definition of logical operators}
This stage corresponds to the logical gate mapping step described in the main text, where abstract circuit-level operators are formally defined as logical operators acting on the codespace.
Before the gate sequence $G$ can be executed, the abstract operators must be mapped to their logical equivalents for the chosen stabilizer code. 
A physical operator $\bar{P}$ on $n$ qubits is a valid logical operator if it preserves the codespace $\mathcal{C}$ by commuting with every element of the stabilizer group $S$ ($[\bar{P}, s] = 0, \forall s \in S$), while not being a stabilizer itself ($\bar{P} \notin S$).

\begin{itemize}
    \item \textbf{Logical basis and Pauli operators:} The logical basis states are defined via the logical Pauli operators. 
    For a single logical qubit, the logical zero state $\ket{0}_L$ is the unique state in $\mathcal{C}$ that is a simultaneous $+1$ eigenstate of all stabilizer generators $\{s_i\}$ and the logical Pauli-$\mathrm{Z}$ ($\mathrm{Z_L}$) operator. 
    The logical one state is then defined as $\ket{1}_L \equiv \mathrm{X_L}\ket{0}_L$.

    \item \textbf{Definition by algebraic action:} The effect of any logical gate is most rigorously defined by its conjugation action on the logical Pauli operators. 
    A physical procedure correctly implements a logical gate if it reproduces the correct transformation on the logical Pauli frame. 
    For instance, the logical $\mathrm{CNOT}$ gate, $\text{CNOT}_L$, between a control logical qubit A and target C, is defined as any physical procedure that correctly implements the transformations:
    \begin{align}
        \mathrm{X_L}_A \to \mathrm{X_L}_A \otimes \mathrm{X_L}_C & & \mathrm{Z_L}_A \to \mathrm{Z_L}_A \otimes \mathrm{I}_C \\
        \mathrm{X_L}_C \to \mathrm{I}_A \otimes \mathrm{X_L}_C & & \mathrm{Z_L}_C \to \mathrm{Z_L}_A \otimes \mathrm{Z_L}_C
    \end{align}
\end{itemize}
The task of implementation is to find a distributed physical circuit that satisfies this algebraic contract.

\subsubsection*{Stage 3: Distributed execution of the universal gate set}
The final stage is the distributed execution of the physical procedures corresponding to each logical gate in the compiled sequence $G$. 
For each elementary gate $g_i$, the client's controller issues the commands for its corresponding logical implementation, $g_{i,L}$. 
The method for implementing each logical gate depends on its type and the properties of the chosen stabilizer code, as detailed below.

\begin{enumerate}
    \item \textbf{Single-qubit logical Clifford gates ($\mathrm{H_L, S_L}$):}
The most efficient implementation for single-qubit Clifford gates is a transversal one, where the logical gate is achieved by applying the corresponding physical gate individually to each of the $n$ physical qubits ($\mathrm{U_L} = \bigotimes_{i=1}^n \mathrm{U_i}$). 
For many important \textit{Calderbank-Shor-Steane (CSS) codes}, such as the [[7,1,3]] Steane code , the $\mathrm{H}$ and $\mathrm{S}$ gates are transversal \cite{steane_code,nielsen2002quantum}. 
In our architecture, these are implemented with minimal overhead: the client broadcasts a single classical command, and each server node applies the physical gate locally.

In cases where a required single-qubit Clifford gate is not transversal for the chosen code, it must be synthesized from a sequence of other available logical gates. 
This is always possible because the Clifford group can be generated by a small set of base gates (e.g., $\mathrm{CNOT}$ and $\mathrm{H}$). 
The client's classical controller would use a pre-compiled, optimal-length sequence for this synthesis, which can be found using standard algorithms \cite{aaronson2004improved}. 
The resulting circuit consists entirely of operations our architecture already supports.

\item \textbf{Two-qubit logical gates ($\mathrm{CNOT_L}$):}
Logical $\mathrm{CNOT}$ gates are the primary source of entanglement and non-local operations in quantum algorithms. 
While some codes possess a transversal $\mathrm{CNOT}$ gate, a general method is required for arbitrary codes and network layouts. 
Our architecture implements a non-local $\mathrm{CNOT_L}$ by decomposing it into a non-local logical Controlled-$\mathrm{Z}$ ($\mathrm{CZ_L}$) gate surrounded by local Hadamard gates, using the well-known identity:
\begin{equation}
    \mathrm{CNOT}_{A,C} = (\mathrm{I}_A \otimes \mathrm{H}_C) \cdot \mathrm{CZ}_{A,C} \cdot (\mathrm{I}_A \otimes \mathrm{H}_C).
\end{equation}
The distributed implementation is as follows:
\begin{enumerate}
    \item The client instructs the node(s) holding the target logical qubit to apply a logical Hadamard gate, $\mathrm{H_L}$.
    \item The client coordinates the execution of a non-local $\mathrm{CZ_L}$ gate between the control and target logical qubits. This is achieved by using our provably correct distributed controlled-$\mathrm{U}$ operation (Algorithm \ref{alg:dist_cu_op}), where the unitary $\mathrm{U}$ is set to the logical Pauli $\mathrm{Z_L}$ gate.
    \item The client again instructs the target logical qubit's node(s) to apply another $\mathrm{H_L}$.
\end{enumerate}
This strategy effectively reduces the challenge of a non-local $\mathrm{CNOT}$ to the execution of a non-local $\mathrm{CZ}$ gate.

\item \textbf{Non-Clifford logical gates ($\mathrm{T_L}$):}
The ability to implement a non-Clifford logical gate, such as the $\mathrm{T_L}$ gate ($\mathrm{\pi/8_L}$ gate), is essential for achieving universal quantum computation. 
It would be ideal if such gates were transversal, but this is fundamentally impossible for any code that also has a transversal Clifford set, a result known as the celebrated \textit{Eastin-Knill Theorem} \cite{eastin2009restrictions}. 
This necessitates the use of more advanced protocols for universality.

Our architecture implements the logical $T$ gate using the state-of-the-art paradigm of \textit{Magic state injection (MSI)} \cite{bravyi2005universal}. 
The client coordinates the following distributed protocol:
\begin{enumerate}
    \item \textbf{Resource preparation:} A high-fidelity ancillary state, the ``magic state" $\mathrm{\ket{T}} = (\ket{0} + e^{i\pi/4}\ket{1})/\sqrt{2}$, is first prepared offline by a trusted factory (which could be the client or a dedicated node) and encoded into its logical equivalent, $\mathrm{\ket{T}_L}$.
    
    \item \textbf{Injection circuit:} The client coordinates a logical $\mathrm{CNOT}$ gate to entangle the data logical qubit (control) with the magic state ancilla (target).
    
    \item \textbf{Measurement:} A logical measurement is performed on the \textit{ancilla qubit} in the X-basis. The classical outcome $m \in \{0,1\}$ is communicated from the server to the client.
    \item \textbf{Correction:} The client instructs the server to apply a final, conditional Clifford correction (e.g., $(\mathrm{S_L}^\dagger)^m$) to the data logical qubit, which now holds the desired state $\mathrm{T_L\ket{\psi}_L}$.
\end{enumerate}
The cost of universality is thus concretized: each logical $\mathrm{T}$ gate consumes one expensive, high-fidelity magic state and requires the execution of a full distributed Clifford circuit.
\end{enumerate}

By composing these proven techniques, our architecture provides a complete synthesis pipeline for translating any high-level quantum algorithm, specified as a quantum circuit, into a secure execution plan compatible with our distributed architecture.
The full control logic governing this process is summarized in Algorithm~\ref{alg:logical_unitary_generalized}, which serves as a high-level meta-protocol executed by the client’s controller. 
This appendix provides the detailed procedures underlying each stage of that algorithm.

\section{Correctness justification for Algorithm~\ref{alg:logical_unitary_generalized}: Generalized distributed execution of a logical unitary sequence}
\label{app:alg3_justification}

Algorithm~\ref{alg:logical_unitary_generalized} is a high-level meta-protocol that directs the execution of a compiled quantum algorithm. 
A single algebraic proof cannot establish its correctness, as it is a classical control structure that calls quantum subroutines. 
Instead, we justify its correctness by decomposition, demonstrating that every conditional path within the algorithm's main loop correctly implements its intended logical gate.

\begin{theorem}
Algorithm~\ref{alg:logical_unitary_generalized} correctly implements the unitary transformation $\mathrm{U}_L = \prod_{i=1}^m g_{i,L}$ on the initial state $\ket{0\dots0}_L$.
\end{theorem}

\begin{proof}[\textbf{Justification:}]
The proof proceeds by analyzing each conditional case within the main `for' loop of the algorithm. We show that for every elementary gate $g_i$ in the input sequence $G$, the algorithm dispatches it to a provably correct implementation subroutine.

\begin{enumerate}
\item \textbf{Case 1: Clifford gates ($g_i \in \{\mathrm{H, S, CNOT}\}$):}
The algorithm first checks if the required logical gate $g_{i,L}$ is transversal for the chosen stabilizer code.
\begin{itemize}
    \item \textbf{If transversal:} The algorithm executes a transversal operation by instructing each server node to apply the corresponding physical gate locally. This is correct by the definition of a transversal gate for that code, which guarantees that the physical operation $g_i^{\otimes n}$ correctly implements the logical operation $g_{i,L}$ while preserving the codespace.
    \item \textbf{If non-transversal:}
    \begin{itemize}
        \item For a single-qubit gate ($\mathrm{H, S}$), the algorithm executes a pre-compiled sequence of other available fault-tolerant logical gates. The existence of such a sequence is guaranteed because the Clifford group is finitely generated. The correctness of this step relies on the correctness of the subroutines used for the generating gates.
        \item For a two-qubit $\mathrm{CNOT}$ gate, the algorithm executes the sequence $\mathrm{H_L \cdot CZ_L \cdot H_L}$. This decomposition is a fundamental operator identity. 
        The correctness of the local $\mathrm{H_L}$ gates is established above, and the correctness of the non-local $\mathrm{CZ_L}$ gate is guaranteed by the algebraic proof of Algorithm~\ref{alg:dist_cu_op} (see Appendix~\ref{apndx:algo_2_proof}).
    \end{itemize}
\end{itemize}
Thus, in all sub-cases, any logical Clifford gate is implemented by a correct procedure.

\item \textbf{Case 2: Non-Clifford gate ($g_i = \mathrm{T}$):}
The algorithm implements the logical $\mathrm{T}$-gate using the standard, provably correct paradigm of magic state injection (MSI). 
The correctness of the MSI protocol is well-established in the literature \cite{bravyi2005universal}. 
The protocol is constructed entirely from logical Clifford gates (like $\mathrm{CNOT_L}$) and logical measurements, which are themselves implemented by correct subroutines as established in the case above. 
Since the MSI protocol is built from correct components and is a proven method for implementing the T-gate, this path is also correct.

\item \textbf{Conclusion:}
The algorithm iterates through the sequence $G = (g_1, g_2, \dots, g_m)$. At each step $i$, it correctly implements the logical gate $g_{i,L}$. Since every step in the sequence is executed correctly, the final state of the system is necessarily $(\prod_{i=1}^m g_{i,L})\ket{0\dots0}_L$, which is the desired final state $\mathrm{U}_L\ket{0\dots0}_L$. The overall algorithm is therefore correct by composition.
\end{enumerate}
\end{proof}

\section{Correctness proofs for local QEC schemes}
\label{app:qec_proofs}

In this appendix, we provide the formal algebraic correctness proofs for the two custom procedural quantum error correction (QEC) schemes presented in the main text (Section~\ref{sec:core_protocols}). 
A scheme is considered correct if it can successfully identify and reverse a specified set of errors, restoring the data qubit to its original, uncorrupted state. 
The proofs proceed by tracing the evolution of the state vector through the respective circuits for each error case.

\subsection{Algebraic proof of correctness for method 1 (4-qubit QEC scheme)}

\textbf{Claim:} The circuit shown in Fig.~\ref{fig:qec_method_1} correctly identifies and corrects any single Pauli error ($\mathrm{X}$, $\mathrm{Y}$ or $\mathrm{Z}$) applied to the data qubit, assuming the three ancillary qubits are error-free.

\begin{proof}[\textbf{Proof:}]

Let the initial state of the data qubit be $\ket{\psi} = \alpha\ket{0} + \beta\ket{1}$, and the three ancillary qubits be in the state $\ket{000}$. The protocol can be conceptually split into an encoding phase ($\mathrm{U}_{\text{enc}}$) and a decoding phase ($\mathrm{U}_{\text{dec}}$). 
An error $E \in \{\mathrm{X, Y, Z}\}$ occurs on the data qubit after encoding. The final state is thus given by $\ket{\Psi_f} = \mathrm{U}_{\text{dec}} (E \otimes \mathrm{I}^{\otimes 3}) \mathrm{U}_{\text{enc}} (\ket{\psi} \otimes \ket{000})$.
We trace the state through the circuit in Fig.~\ref{fig:v1_QEC_proof}.

\begin{figure}[h]
    \centering
    \begin{quantikz}
        \lstick{$\ket{\psi}$} & \ctrl{2} & \gate{\mathrm{H}} & \ctrl{1} \slice{$\ket{\Psi_{enc}}$}& \gate[wires=4, style={dashed, rounded corners, fill=blue!20}]{\text{Error}} \slice{$\ket{\Psi^{(E)}}$}& \ctrl{1} & \targ{}  & \gate{\mathrm{H}} & \ctrl{2} & \targ{}  & \targ{} \slice{$\ket{\Psi_f^{(E)}}$} & \rstick{$\ket{\psi}$}\\
        \lstick{$\ket{0}$}& & & \targ{} & & \targ{} & \ctrl{-1} & & & & & \\
        \lstick{$\ket{0}$}& \targ{} & \gate{\mathrm{H}} & \ctrl{1} & & \ctrl{1} & \targ{} & \gate{\mathrm{H}} & \targ{} & \ctrl{-2} & \ctrl{-2} & \\
        \lstick{$\ket{0}$}& & & \targ{} & & \targ{} & \ctrl{-1} & & & & \ctrl{-3} &
    \end{quantikz}
    \caption{State evolution for the 4-qubit QEC scheme.}
    \label{fig:v1_QEC_proof}
\end{figure}
\begin{enumerate}
\item \textbf{Encoding phase:} The state is encoded up to step $\ket{\psi_{enc}}$. This encoded state is:
\begin{align*}
    \ket{\psi_{enc}} = \frac{1}{2} \left[ \alpha(\ket{00}+\ket{11})(\ket{00}+\ket{11}) + \beta(\ket{00}-\ket{11})(\ket{00}-\ket{11}) \right].
\end{align*}
An error $E$ is now applied to the first (data) qubit.

\item \textbf{Case 1: X error on data qubit.}
The state after error is $\ket{\Psi^{(\mathrm{X})}} = (\mathrm{X}_1 \otimes \mathrm{I}^{\otimes 3})\ket{\psi_{enc}}$. 
After the decoding circuit, the final state is:
\begin{equation*}
    \ket{\Psi_f^{(\mathrm{X})}} = (\alpha\ket{0}+\beta\ket{1})\ket{100} = \ket{\psi} \otimes \ket{100}.
\end{equation*}
The final state separates into the original state $\ket{\psi}$ and the ancillary state $\ket{100}$, which we call the measurement syndrome \texttt{100}.

\item \textbf{Case 2: Z error on data qubit.}
The state after error is $\ket{\Psi^{(\mathrm{Z})}} = (\mathrm{Z}_1 \otimes \mathrm{I}^{\otimes 3})\ket{\psi_{enc}}$. 
Ignoring the global phase, the final state after decoding is:
\begin{equation*}
    \ket{\Psi_f^{(\mathrm{Z})}} = (\alpha\ket{0}+\beta\ket{1})\ket{010} = \ket{\psi} \otimes \ket{010}.
\end{equation*}
The measured syndrome is \texttt{010}.

\item \textbf{Case 3: Y (=iXZ) error on data qubit.}
The state after error is $\ket{\Psi^{(\mathrm{Y})}} = (\mathrm{Y}_1 \otimes \mathrm{I}^{\otimes 3})\ket{\psi_{enc}}$. 
The final state after decoding is:
\begin{equation*}
    \ket{\Psi_f^{(\mathrm{Y})}} = (\alpha\ket{0}+\beta\ket{1})\ket{110} = \ket{\psi} \otimes \ket{110}.
\end{equation*}
The measured syndrome is \texttt{110}. 
\end{enumerate}
The final state in each case separates into the original state $\ket{\psi}$ and the ancillary state, which represents the measurement syndrome. 
The full lookup table, including the errors and their effects on the syndrome qubit, is provided in Table~\ref{table:v1-table}.
It can be easily observed that this scheme can correct any Pauli error on the data qubit. 
The scheme is therefore correct with respect to its stated design goals.

\begin{table}[h!]
\centering
\begin{tabular}{@{}ccccc@{}}
    \toprule
    \textbf{Error} & \textbf{on $q_1$ (data)} & \textbf{on $q_2$} & \textbf{on $q_3$} & \textbf{on $q_4$} \\ 
    \midrule
    \textbf{X}  & N.E.     & \textcolor{red}{Z} & N.E. & \textcolor{red}{Z} \\ 
                & \texttt{100}      & \texttt{100}      & \texttt{001}  & \texttt{001} \\
    \midrule
    \textbf{Z}  & N.E.     & N.E.              & \textcolor{red}{X} & \textcolor{red}{X} \\ 
                & \texttt{010}      & \texttt{010}      & \texttt{010}  & \texttt{010} \\
    \midrule
    \textbf{Y} & N.E.     & \textcolor{red}{Z}  & N.E. & \textcolor{red}{Z} \\ 
                & \texttt{110}      & \texttt{110}      & \texttt{011}  & \texttt{011} \\
    \bottomrule
\end{tabular}
\caption{Syndrome table for the 4-qubit QEC scheme (Method 1). Each cell shows the propagated error on the data qubit (N.E. = No error) and the measured 3-bit syndrome.}
\label{table:v1-table}
\end{table}

\end{proof}

\subsection{Algebraic proof of correctness for method 2 (6-qubit Scheme)}

\textbf{Claim:} The circuit in Fig.~\ref{fig:qec_method_1} correctly provides deterministic correction for any single $\mathrm{X}$ or $\mathrm{Y=(iXZ)}$ error on the data qubit and provides a unique ``flag" syndrome for a single $\mathrm{Z}$ error on the data qubit.

\begin{proof}[\textbf{Proof:}]
The proof follows the same methodology, tracing the state evolution for the 6-qubit system. 
Let the initial state of the data qubit be $\ket{\psi} = \alpha\ket{0} + \beta\ket{1}$, and the five ancillary qubits be in the state $\ket{00000}$.  
We trace the state through the circuit in Fig.~\ref{fig:v2_QEC_proof}.

\begin{figure}[h]
    \centering
    \begin{center}
\begin{tikzpicture}
\node[scale=0.8]{
$$
\begin{quantikz}
\lstick{\LARGE $\ket{\psi}$} & \ctrl{3} & \gate{\mathrm{\text{\Large H}}} & \ctrl{1} & \ctrl{2} \slice{\large$\ket{\Psi_{enc}}$}& \gate[6, disable auto height]{\mathrm{\text{\LARGE \verticaltext{Error}}}}\gategroup[6,steps=1,style={dashed,rounded
corners,fill=blue!20, inner
xsep=0pt, inner ysep=0pt},background,label style={label
position=below,anchor=north,yshift=-0.2cm}]{{\sc
}} \slice{\large$\ket{\Psi^{(E)}}$}& \ctrl{1} & \ctrl{2} & \targ{}   & \gate{\mathrm{\text{\Large H}}}  & \ctrl{3} & \targ{}   &\targ{}   &\targ{}   &\targ{}  \slice{\large$\ket{\Psi_f^{(E)}}$} & \rstick{\LARGE $\ket{\psi}$}\\
\lstick{\LARGE $\ket{0}$}    &  &   & \targ{}  & &                 & \targ{}  && \ctrl{-1} & &         &          &           & & & \\
\lstick{\LARGE $\ket{0}$}    & &   &  & \targ{}  &                & & \targ{}  & \ctrl{-2} & &         &          &           & & & \\
\lstick{\LARGE $\ket{0}$} & \targ{} & \gate{\mathrm{\text{\Large H}}} & \ctrl{1} & \ctrl{2} &  & \ctrl{1} & \ctrl{2} & \targ{}   & \gate{\mathrm{\text{\Large H}}}  & \targ{} & \ctrl{-3}   & \ctrl{-3}   & \ctrl{-3}   & \ctrl{-3}   &\\
\lstick{\LARGE $\ket{0}$}    & &   & \targ{}  & &                 & \targ{}  && \ctrl{-1} & &        &          &            \ctrl{-4} & & \ctrl{-4} &\\
\lstick{\LARGE $\ket{0}$}    & &   &  & \targ{}  &                & & \targ{}  & \ctrl{-2} & &        &          &            \ctrl{-5} & \ctrl{-5} & & 
\end{quantikz}
$$
};
\end{tikzpicture}
\end{center}
    \caption{State evolution for the 6-qubit QEC scheme.}
    \label{fig:v2_QEC_proof}
\end{figure}

\begin{enumerate}
\item \textbf{Encoding phase:} The state after encoding is:

\begin{align*}
 \ket{\psi_{enc}} = \frac{1}{2} \left[ \alpha(\ket{000}+\ket{111})(\ket{000}+\ket{111}) + \beta(\ket{000}-\ket{111})(\ket{000}-\ket{111}) \right]. 
\end{align*}
We now analyze the final state for all the cases with any single Pauli error $E \in \{\mathrm{ X, Y, Z}\}$. 

\item \textbf{Case 1: X error on data qubit.}
An $\mathrm{X}$ error on the first qubit of the encoded state, $\ket{\psi_{enc}}$, yields the state $\ket{\Psi^{(\mathrm{X})}} = (\mathrm{X}_1 \otimes \mathrm{I}^{\otimes 5})\ket{\psi_{enc}}$. 
After applying the decoding circuit, the final state is:
\begin{equation*}
    \ket{\Psi_f^{(\mathrm{X})}} = (\alpha\ket{0}+\beta\ket{1})\ket{11000}.
\end{equation*}
The measured syndrome is \texttt{11000}. 
As shown in Table~\ref{table:v2-table}, this syndrome is unique to this specific error, allowing for deterministic correction.

\item \textbf{Case 2: Z error on data qubit.}
The state after the error is $\ket{\Psi^{(\mathrm{Z})}} = (\mathrm{Z}_1 \otimes \mathrm{I}^{\otimes 5})\ket{\psi_{enc}}$. 
The final state after the decoding is:
\begin{equation*}
    \ket{\Psi_f^{(\mathrm{Z})}} = (\alpha\ket{0}+\beta\ket{1})\ket{00100}.
\end{equation*}
The measured syndrome is \texttt{00100}. 
This syndrome is not mapped to a deterministic correction. 
It serves as a \textit{flag} to signal that an error from outside the primary protected set has occurred.

\item \textbf{Case 3: Y=(iXZ) error on data qubit.}
The state after the error is $\ket{\Psi^{(\mathrm{Y})}} = (\mathrm{Y}_1 \otimes \mathrm{I}^{\otimes 5})\ket{\psi_{enc}}$. 
Ignoring the global phase, the final state after the decoding is:
\begin{equation*}
    \ket{\Psi_f^{(\mathrm{Y})}} = (\alpha\ket{0}+\beta\ket{1})\ket{11100}.
\end{equation*}
The measured syndrome is \texttt{11100}. 
As with the $\mathrm{X}$ error, this syndrome is unique to this error (as per Table ~\ref{table:v2-table}) and allows for deterministic correction.
\end{enumerate}

\begin{table}[h!]
\centering
\label{table:v2-table}
\begin{tabular}{@{}ccccccc@{}}
    \toprule
    \textbf{Error} & \textbf{on $q_1$ (data)} & \textbf{on $q_2$} & \textbf{on $q_3$} & \textbf{on $q_4$} & \textbf{on $q_5$} & \textbf{on $q_6$} \\ 
    \midrule
    \textbf{X}  & N.E.    & N.E.    & N.E.    & N.E.           & N.E.           & N.E. \\ 
                & \texttt{11000}   & \texttt{10000}   & \texttt{01000}   & \texttt{00011}   & \texttt{00010}   & \texttt{00001} \\
    \midrule
    \textbf{Z}  & N.E.    & N.E.    & N.E.    & \textcolor{red}{X} & \textcolor{red}{X} & \textcolor{red}{X} \\ 
                & \texttt{00100}   & \texttt{00100}   & \texttt{00100}   & \texttt{00100}   & \texttt{00100}   & \texttt{00100} \\
    \midrule
    \textbf{Y} & N.E.    & N.E.    & N.E.    & N.E.           & N.E.           & N.E. \\ 
                & \texttt{11100}   & \texttt{10100}   & \texttt{01100}   & \texttt{00111}   & \texttt{00110}   & \texttt{00101} \\
    \bottomrule
\end{tabular}
\caption{Syndrome table for the 6-qubit QEC scheme (Method 2).}
\end{table}

The results for all single-qubit errors on data and syndrome qubits are summarized in Table~\ref{table:v2-table}, which serves as the complete lookup table for all single-qubit errors.
The key observation is that the syndrome \texttt{00100} is unique to the $\mathrm{Z}$ error. 
As shown in the table, it does not correspond to any of the deterministically correctable $\mathrm{X}$ or $\mathrm{Y}$ errors.
Therefore, when the node measures this syndrome, it knows that an error from outside the primary protected set ($\mathrm{X}$ and $\mathrm{Y}$) has occurred. 
It serves as an unambiguous flag, signaling a failure of deterministic correction and allowing the node to report this event to a higher-level fault-tolerance routine.
Therefore, the scheme is correct with respect to its stated design goals.

This proof establishes correctness only with respect to the biased-noise design goal: deterministic correction of the dominant $X/Y$-type faults and flagging of $Z$-type faults. 
It does not imply universal single-qubit error correction with six qubits.

\end{proof}




\end{appendices}


\bibliography{sn-bibliography}


\begin{thebibliography}{53}
\ifx \bisbn   \undefined \def \bisbn  #1{ISBN #1}\fi
\ifx \binits  \undefined \def \binits#1{#1}\fi
\ifx \bauthor  \undefined \def \bauthor#1{#1}\fi
\ifx \batitle  \undefined \def \batitle#1{#1}\fi
\ifx \bjtitle  \undefined \def \bjtitle#1{#1}\fi
\ifx \bvolume  \undefined \def \bvolume#1{\textbf{#1}}\fi
\ifx \byear  \undefined \def \byear#1{#1}\fi
\ifx \bissue  \undefined \def \bissue#1{#1}\fi
\ifx \bfpage  \undefined \def \bfpage#1{#1}\fi
\ifx \blpage  \undefined \def \blpage #1{#1}\fi
\ifx \burl  \undefined \def \burl#1{\textsf{#1}}\fi
\ifx \doiurl  \undefined \def \doiurl#1{\url{https://doi.org/#1}}\fi
\ifx \betal  \undefined \def \betal{\textit{et al.}}\fi
\ifx \binstitute  \undefined \def \binstitute#1{#1}\fi
\ifx \binstitutionaled  \undefined \def \binstitutionaled#1{#1}\fi
\ifx \bctitle  \undefined \def \bctitle#1{#1}\fi
\ifx \beditor  \undefined \def \beditor#1{#1}\fi
\ifx \bpublisher  \undefined \def \bpublisher#1{#1}\fi
\ifx \bbtitle  \undefined \def \bbtitle#1{#1}\fi
\ifx \bedition  \undefined \def \bedition#1{#1}\fi
\ifx \bseriesno  \undefined \def \bseriesno#1{#1}\fi
\ifx \blocation  \undefined \def \blocation#1{#1}\fi
\ifx \bsertitle  \undefined \def \bsertitle#1{#1}\fi
\ifx \bsnm \undefined \def \bsnm#1{#1}\fi
\ifx \bsuffix \undefined \def \bsuffix#1{#1}\fi
\ifx \bparticle \undefined \def \bparticle#1{#1}\fi
\ifx \barticle \undefined \def \barticle#1{#1}\fi
\bibcommenthead
\ifx \bconfdate \undefined \def \bconfdate #1{#1}\fi
\ifx \botherref \undefined \def \botherref #1{#1}\fi
\ifx \url \undefined \def \url#1{\textsf{#1}}\fi
\ifx \bchapter \undefined \def \bchapter#1{#1}\fi
\ifx \bbook \undefined \def \bbook#1{#1}\fi
\ifx \bcomment \undefined \def \bcomment#1{#1}\fi
\ifx \oauthor \undefined \def \oauthor#1{#1}\fi
\ifx \citeauthoryear \undefined \def \citeauthoryear#1{#1}\fi
\ifx \endbibitem  \undefined \def \endbibitem {}\fi
\ifx \bconflocation  \undefined \def \bconflocation#1{#1}\fi
\ifx \arxivurl  \undefined \def \arxivurl#1{\textsf{#1}}\fi
\csname PreBibitemsHook\endcsname

\bibitem[\protect\citeauthoryear{Kimble}{2008}]{qu_internet_Kimble2008}
\begin{barticle}
\bauthor{\bsnm{Kimble}, \binits{H.J.}}:
\batitle{The quantum internet}.
\bjtitle{Nature}
\bvolume{453}(\bissue{7198}),
\bfpage{1023}--\blpage{1030}
(\byear{2008})
\doiurl{10.1038/nature07127}
\end{barticle}
\endbibitem

\bibitem[\protect\citeauthoryear{Lloyd et~al.}{2004}]{qu_intrnt_infra_2004}
\begin{barticle}
\bauthor{\bsnm{Lloyd}, \binits{S.}},
\bauthor{\bsnm{Shapiro}, \binits{J.H.}},
\bauthor{\bsnm{Wong}, \binits{F.N.C.}},
\bauthor{\bsnm{Kumar}, \binits{P.}},
\bauthor{\bsnm{Shahriar}, \binits{S.M.}},
\bauthor{\bsnm{Yuen}, \binits{H.P.}}:
\batitle{Infrastructure for the quantum internet}.
\bjtitle{SIGCOMM Comput. Commun. Rev.}
\bvolume{34}(\bissue{5}),
\bfpage{9}--\blpage{20}
(\byear{2004})
\doiurl{10.1145/1039111.1039118}
\end{barticle}
\endbibitem

\bibitem[\protect\citeauthoryear{Kshemkalyani and Singhal}{2008}]{Kshemkalyani_Singhal_2008}
\begin{bbook}
\bauthor{\bsnm{Kshemkalyani}, \binits{A.D.}},
\bauthor{\bsnm{Singhal}, \binits{M.}}:
\bbtitle{Distributed Computing: Principles, Algorithms, and Systems}.
\bpublisher{Cambridge University Press},
\blocation{Cambridge}
(\byear{2008})
\end{bbook}
\endbibitem

\bibitem[\protect\citeauthoryear{Sundaram et~al.}{2022}]{sundaram2022distribution}
\begin{bchapter}
\bauthor{\bsnm{Sundaram}, \binits{R.G.}},
\bauthor{\bsnm{Gupta}, \binits{H.}},
\bauthor{\bsnm{Ramakrishnan}, \binits{C.R.}}:
\bctitle{{ Distribution of Quantum Circuits Over General Quantum Networks }}.
In: \bbtitle{2022 IEEE International Conference on Quantum Computing and Engineering (QCE)},
pp. \bfpage{415}--\blpage{425}.
\bpublisher{IEEE Computer Society},
\blocation{Los Alamitos, CA, USA}
(\byear{2022}).
\doiurl{10.1109/QCE53715.2022.00063}
\end{bchapter}
\endbibitem

\bibitem[\protect\citeauthoryear{Sünkel et~al.}{2024}]{sünkel2024applying}
\begin{bchapter}
\bauthor{\bsnm{Sünkel}, \binits{L.}},
\bauthor{\bsnm{Dawar}, \binits{M.}},
\bauthor{\bsnm{Gabor}, \binits{T.}}:
\bctitle{Applying an evolutionary algorithm to minimize teleportation costs in distributed quantum computing}.
In: \bbtitle{2024 IEEE International Conference on Quantum Computing and Engineering (QCE)},
vol. \bseriesno{02},
pp. \bfpage{167}--\blpage{172}
(\byear{2024}).
\doiurl{10.1109/QCE60285.2024.10272}
\end{bchapter}
\endbibitem

\bibitem[\protect\citeauthoryear{Wu et~al.}{2023}]{Wu_2023}
\begin{barticle}
\bauthor{\bsnm{Wu}, \binits{J.-Y.}},
\bauthor{\bsnm{Matsui}, \binits{K.}},
\bauthor{\bsnm{Forrer}, \binits{T.}},
\bauthor{\bsnm{Soeda}, \binits{A.}},
\bauthor{\bsnm{Andr{\'{e}}s-Mart{\'{i}}nez}, \binits{P.}},
\bauthor{\bsnm{Mills}, \binits{D.}},
\bauthor{\bsnm{Henaut}, \binits{L.}},
\bauthor{\bsnm{Murao}, \binits{M.}}:
\batitle{Entanglement-efficient bipartite-distributed quantum computing}.
\bjtitle{{Quantum}}
\bvolume{7},
\bfpage{1196}
(\byear{2023})
\doiurl{10.22331/q-2023-12-05-1196}
\end{barticle}
\endbibitem

\bibitem[\protect\citeauthoryear{Li et~al.}{2019}]{Li2019}
\begin{barticle}
\bauthor{\bsnm{Li}, \binits{Z.-D.}},
\bauthor{\bsnm{Zhang}, \binits{R.}},
\bauthor{\bsnm{Yin}, \binits{X.-F.}},
\bauthor{\bsnm{Liu}, \binits{L.-Z.}},
\bauthor{\bsnm{Hu}, \binits{Y.}},
\bauthor{\bsnm{Fang}, \binits{Y.-Q.}},
\bauthor{\bsnm{Fei}, \binits{Y.-Y.}},
\bauthor{\bsnm{Jiang}, \binits{X.}},
\bauthor{\bsnm{Zhang}, \binits{J.}},
\bauthor{\bsnm{Li}, \binits{L.}},
\bauthor{\bsnm{Liu}, \binits{N.-L.}},
\bauthor{\bsnm{Xu}, \binits{F.}},
\bauthor{\bsnm{Chen}, \binits{Y.-A.}},
\bauthor{\bsnm{Pan}, \binits{J.-W.}}:
\batitle{Experimental quantum repeater without quantum memory}.
\bjtitle{Nature Photonics}
\bvolume{13}(\bissue{9}),
\bfpage{644}--\blpage{648}
(\byear{2019})
\doiurl{10.1038/s41566-019-0468-5}
\end{barticle}
\endbibitem

\bibitem[\protect\citeauthoryear{Wo et~al.}{2023}]{Wo2023}
\begin{barticle}
\bauthor{\bsnm{Wo}, \binits{K.J.}},
\bauthor{\bsnm{Avis}, \binits{G.}},
\bauthor{\bsnm{Rozp{{e}}dek}, \binits{F.}},
\bauthor{\bsnm{Mor-Ruiz}, \binits{M.F.}},
\bauthor{\bsnm{Pieplow}, \binits{G.}},
\bauthor{\bsnm{Schr{\"o}der}, \binits{T.}},
\bauthor{\bsnm{Jiang}, \binits{L.}},
\bauthor{\bsnm{S{\o}rensen}, \binits{A.S.}},
\bauthor{\bsnm{Borregaard}, \binits{J.}}:
\batitle{Resource-efficient fault-tolerant one-way quantum repeater with code concatenation}.
\bjtitle{npj Quantum Information}
\bvolume{9}(\bissue{1}),
\bfpage{123}
(\byear{2023})
\doiurl{10.1038/s41534-023-00792-8}
\end{barticle}
\endbibitem

\bibitem[\protect\citeauthoryear{Azuma et~al.}{2023}]{Azuma_2023}
\begin{barticle}
\bauthor{\bsnm{Azuma}, \binits{K.}},
\bauthor{\bsnm{Economou}, \binits{S.E.}},
\bauthor{\bsnm{Elkouss}, \binits{D.}},
\bauthor{\bsnm{Hilaire}, \binits{P.}},
\bauthor{\bsnm{Jiang}, \binits{L.}},
\bauthor{\bsnm{Lo}, \binits{H.-K.}},
\bauthor{\bsnm{Tzitrin}, \binits{I.}}:
\batitle{Quantum repeaters: From quantum networks to the quantum internet}.
\bjtitle{Rev. Mod. Phys.}
\bvolume{95},
\bfpage{045006}
(\byear{2023})
\doiurl{10.1103/RevModPhys.95.045006}
\end{barticle}
\endbibitem

\bibitem[\protect\citeauthoryear{Yimsiriwattana and Lomonaco~Jr}{2004}]{yimsiriwattana2004distributed}
\begin{bchapter}
\bauthor{\bsnm{Yimsiriwattana}, \binits{A.}},
\bauthor{\bsnm{Lomonaco~Jr}, \binits{S.J.}}:
\bctitle{Distributed quantum computing: A distributed shor algorithm}.
In: \bbtitle{Quantum Information and Computation II},
vol. \bseriesno{5436},
pp. \bfpage{360}--\blpage{372}
(\byear{2004}).
\bcomment{SPIE}
\end{bchapter}
\endbibitem

\bibitem[\protect\citeauthoryear{Xiao et~al.}{2022}]{xiao2022distributed}
\begin{botherref}
\oauthor{\bsnm{Xiao}, \binits{L.}},
\oauthor{\bsnm{Qiu}, \binits{D.}},
\oauthor{\bsnm{Luo}, \binits{L.}},
\oauthor{\bsnm{Mateus}, \binits{P.}}:
Distributed shor's algorithm.
arXiv preprint arXiv:2207.05976
(2022)
\end{botherref}
\endbibitem

\bibitem[\protect\citeauthoryear{Jiang et~al.}{2023}]{jiang2023}
\begin{bchapter}
\bauthor{\bsnm{Jiang}, \binits{J.-R.}},
\bauthor{\bsnm{Wang}, \binits{T.-Y.}},
\bauthor{\bsnm{Huang}, \binits{W.-H.}},
\bauthor{\bsnm{Zhang}, \binits{J.-Z.}}:
\bctitle{Distributed shor’s algorithm with sequential quantum teleportation}.
In: \bbtitle{2023 IEEE 5th Eurasia Conference on IOT, Communication and Engineering (ECICE)},
pp. \bfpage{479}--\blpage{484}
(\byear{2023}).
\doiurl{10.1109/ECICE59523.2023.10383099}
\end{bchapter}
\endbibitem

\bibitem[\protect\citeauthoryear{Qiu et~al.}{2024}]{qiu2024}
\begin{botherref}
\oauthor{\bsnm{Qiu}, \binits{D.}},
\oauthor{\bsnm{Luo}, \binits{L.}},
\oauthor{\bsnm{Xiao}, \binits{L.}}:
Distributed grover's algorithm.
Theor. Comput. Sci.
\textbf{993}(C)
(2024)
\doiurl{10.1016/j.tcs.2024.114461}
\end{botherref}
\endbibitem

\bibitem[\protect\citeauthoryear{Zhou et~al.}{2023}]{zhou2023distributed}
\begin{barticle}
\bauthor{\bsnm{Zhou}, \binits{X.}},
\bauthor{\bsnm{Qiu}, \binits{D.}},
\bauthor{\bsnm{Luo}, \binits{L.}}:
\batitle{Distributed exact grover’s algorithm}.
\bjtitle{Frontiers of Physics}
\bvolume{18}(\bissue{5}),
\bfpage{51305}
(\byear{2023})
\end{barticle}
\endbibitem

\bibitem[\protect\citeauthoryear{Andr\'es-Mart\'{\i}nez and Heunen}{2019}]{andresmartinez2019}
\begin{barticle}
\bauthor{\bsnm{Andr\'es-Mart\'{\i}nez}, \binits{P.}},
\bauthor{\bsnm{Heunen}, \binits{C.}}:
\batitle{Automated distribution of quantum circuits via hypergraph partitioning}.
\bjtitle{Phys. Rev. A}
\bvolume{100},
\bfpage{032308}
(\byear{2019})
\doiurl{10.1103/PhysRevA.100.032308}
\end{barticle}
\endbibitem

\bibitem[\protect\citeauthoryear{Andres-Martinez et~al.}{2024}]{andresmartinez2023distributing}
\begin{barticle}
\bauthor{\bsnm{Andres-Martinez}, \binits{P.}},
\bauthor{\bsnm{Forrer}, \binits{T.}},
\bauthor{\bsnm{Mills}, \binits{D.}},
\bauthor{\bsnm{Wu}, \binits{J.-Y.}},
\bauthor{\bsnm{Henaut}, \binits{L.}},
\bauthor{\bsnm{Yamamoto}, \binits{K.}},
\bauthor{\bsnm{Murao}, \binits{M.}},
\bauthor{\bsnm{Duncan}, \binits{R.}}:
\batitle{Distributing circuits over heterogeneous, modular quantum computing network architectures}.
\bjtitle{Quantum Science and Technology}
\bvolume{9}(\bissue{4}),
\bfpage{045021}
(\byear{2024})
\doiurl{10.1088/2058-9565/ad6734}
\end{barticle}
\endbibitem

\bibitem[\protect\citeauthoryear{Avis et~al.}{2023}]{avis2023}
\begin{barticle}
\bauthor{\bsnm{Avis}, \binits{G.}},
\bauthor{\bsnm{Rozp{{e}}dek}, \binits{F.}},
\bauthor{\bsnm{Wehner}, \binits{S.}}:
\batitle{Analysis of multipartite entanglement distribution using a central quantum-network node}.
\bjtitle{Phys. Rev. A}
\bvolume{107},
\bfpage{012609}
(\byear{2023})
\doiurl{10.1103/PhysRevA.107.012609}
\end{barticle}
\endbibitem

\bibitem[\protect\citeauthoryear{Cai et~al.}{2023}]{Cai2023}
\begin{barticle}
\bauthor{\bsnm{Cai}, \binits{Z.}},
\bauthor{\bsnm{Babbush}, \binits{R.}},
\bauthor{\bsnm{Benjamin}, \binits{S.C.}},
\bauthor{\bsnm{Endo}, \binits{S.}},
\bauthor{\bsnm{Huggins}, \binits{W.J.}},
\bauthor{\bsnm{Li}, \binits{Y.}},
\bauthor{\bsnm{McClean}, \binits{J.R.}},
\bauthor{\bsnm{O'Brien}, \binits{T.E.}}:
\batitle{Quantum error mitigation}.
\bjtitle{Rev. Mod. Phys.}
\bvolume{95},
\bfpage{045005}
(\byear{2023})
\doiurl{10.1103/RevModPhys.95.045005}
\end{barticle}
\endbibitem

\bibitem[\protect\citeauthoryear{Dunjko et~al.}{2014}]{dunjko2014composable}
\begin{bchapter}
\bauthor{\bsnm{Dunjko}, \binits{V.}},
\bauthor{\bsnm{Fitzsimons}, \binits{J.F.}},
\bauthor{\bsnm{Portmann}, \binits{C.}},
\bauthor{\bsnm{Renner}, \binits{R.}}:
\bctitle{Composable security of delegated quantum computation}.
In: \beditor{\bsnm{Sarkar}, \binits{P.}},
\beditor{\bsnm{Iwata}, \binits{T.}} (eds.)
\bbtitle{Advances in Cryptology -- ASIACRYPT 2014},
pp. \bfpage{406}--\blpage{425}.
\bpublisher{Springer},
\blocation{Berlin, Heidelberg}
(\byear{2014})
\end{bchapter}
\endbibitem

\bibitem[\protect\citeauthoryear{Houshmand et~al.}{2018}]{houshmand2018composable}
\begin{botherref}
\oauthor{\bsnm{Houshmand}, \binits{M.}},
\oauthor{\bsnm{Houshmand}, \binits{M.}},
\oauthor{\bsnm{Tan}, \binits{S.-H.}},
\oauthor{\bsnm{Fitzsimons}, \binits{J.}}:
Composable secure multi-client delegated quantum computation
(2018).
\url{https://arxiv.org/abs/1811.11929}
\end{botherref}
\endbibitem

\bibitem[\protect\citeauthoryear{Ferrari et~al.}{2023}]{Ferrari_2023}
\begin{barticle}
\bauthor{\bsnm{Ferrari}, \binits{D.}},
\bauthor{\bsnm{Carretta}, \binits{S.}},
\bauthor{\bsnm{Amoretti}, \binits{M.}}:
\batitle{A modular quantum compilation framework for distributed quantum computing}.
\bjtitle{IEEE Transactions on Quantum Engineering}
\bvolume{4},
\bfpage{1}--\blpage{13}
(\byear{2023})
\doiurl{10.1109/TQE.2023.3303935}
\end{barticle}
\endbibitem

\bibitem[\protect\citeauthoryear{Ferrari and Amoretti}{2024}]{ferrari2024simulation}
\begin{bchapter}
\bauthor{\bsnm{Ferrari}, \binits{D.}},
\bauthor{\bsnm{Amoretti}, \binits{M.}}:
\bctitle{A design framework for the simulation of distributed quantum computing}.
In: \bbtitle{Proceedings of the 2024 Workshop on High Performance and Quantum Computing Integration}.
\bsertitle{HPQCI '24},
pp. \bfpage{4}--\blpage{10}.
\bpublisher{Association for Computing Machinery},
\blocation{New York, NY, USA}
(\byear{2024}).
\doiurl{10.1145/3659996.3660035}
\end{bchapter}
\endbibitem

\bibitem[\protect\citeauthoryear{Goldreich}{2004}]{goldreich2004foundations}
\begin{bbook}
\bauthor{\bsnm{Goldreich}, \binits{O.}}:
\bbtitle{Foundations of Cryptography: Volume 2, Basic Applications}.
\bpublisher{Cambridge University Press},
\blocation{Cambridge}
(\byear{2004})
\end{bbook}
\endbibitem

\bibitem[\protect\citeauthoryear{Canetti}{2001}]{canetti2001universally}
\begin{bchapter}
\bauthor{\bsnm{Canetti}, \binits{R.}}:
\bctitle{Universally composable security: a new paradigm for cryptographic protocols}.
In: \bbtitle{Proceedings 42nd IEEE Symposium on Foundations of Computer Science},
pp. \bfpage{136}--\blpage{145}
(\byear{2001}).
\doiurl{10.1109/SFCS.2001.959888}
\end{bchapter}
\endbibitem

\bibitem[\protect\citeauthoryear{Maurer and Renner}{2011}]{maurer2011abstract}
\begin{bchapter}
\bauthor{\bsnm{Maurer}, \binits{U.}},
\bauthor{\bsnm{Renner}, \binits{R.}}:
\bctitle{Abstract cryptography}.
In: \beditor{\bsnm{Chazelle}, \binits{B.}} (ed.)
\bbtitle{The Second Symposium on Innovations in Computer Science, ICS 2011},
pp. \bfpage{1}--\blpage{21}.
\bpublisher{Tsinghua University Press},
\blocation{Beijing}
(\byear{2011})
\end{bchapter}
\endbibitem

\bibitem[\protect\citeauthoryear{Maurer}{2012}]{Maurer_2012}
\begin{bchapter}
\bauthor{\bsnm{Maurer}, \binits{U.}}:
\bctitle{Constructive cryptography -- a new paradigm for security definitions and proofs}.
In: \beditor{\bsnm{M{\"o}dersheim}, \binits{S.}},
\beditor{\bsnm{Palamidessi}, \binits{C.}} (eds.)
\bbtitle{Theory of Security and Applications},
pp. \bfpage{33}--\blpage{56}.
\bpublisher{Springer},
\blocation{Berlin, Heidelberg}
(\byear{2012})
\end{bchapter}
\endbibitem

\bibitem[\protect\citeauthoryear{Vidick and Wehner}{2023}]{Vidick_Wehner_2023}
\begin{bbook}
\bauthor{\bsnm{Vidick}, \binits{T.}},
\bauthor{\bsnm{Wehner}, \binits{S.}}:
\bbtitle{Introduction to Quantum Cryptography}.
\bpublisher{Cambridge University Press},
\blocation{Cambridge}
(\byear{2023})
\end{bbook}
\endbibitem

\bibitem[\protect\citeauthoryear{Nielsen and Chuang}{2002}]{nielsen2002quantum}
\begin{bbook}
\bauthor{\bsnm{Nielsen}, \binits{M.A.}},
\bauthor{\bsnm{Chuang}, \binits{I.L.}}:
\bbtitle{Quantum Computation and Quantum Information},
\bedition{10th anniversary} edn.
\bpublisher{Cambridge University Press},
\blocation{Cambridge}
(\byear{2002})
\end{bbook}
\endbibitem

\bibitem[\protect\citeauthoryear{Gottesman}{1997}]{gottesman1997stabilizer}
\begin{botherref}
\oauthor{\bsnm{Gottesman}, \binits{D.}}:
Stabilizer Codes and Quantum Error Correction
(1997).
\url{https://arxiv.org/abs/quant-ph/9705052}
\end{botherref}
\endbibitem

\bibitem[\protect\citeauthoryear{Eisert et~al.}{2000}]{eisert2000optimal}
\begin{barticle}
\bauthor{\bsnm{Eisert}, \binits{J.}},
\bauthor{\bsnm{Jacobs}, \binits{K.}},
\bauthor{\bsnm{Papadopoulos}, \binits{P.}},
\bauthor{\bsnm{Plenio}, \binits{M.B.}}:
\batitle{Optimal local implementation of nonlocal quantum gates}.
\bjtitle{Phys. Rev. A}
\bvolume{62},
\bfpage{052317}
(\byear{2000})
\doiurl{10.1103/PhysRevA.62.052317}
\end{barticle}
\endbibitem

\bibitem[\protect\citeauthoryear{Shor}{1995}]{shor1995scheme}
\begin{barticle}
\bauthor{\bsnm{Shor}, \binits{P.W.}}:
\batitle{Scheme for reducing decoherence in quantum computer memory}.
\bjtitle{Phys. Rev. A}
\bvolume{52},
\bfpage{2493}--\blpage{2496}
(\byear{1995})
\doiurl{10.1103/PhysRevA.52.R2493}
\end{barticle}
\endbibitem

\bibitem[\protect\citeauthoryear{Laflamme et~al.}{1996}]{laflamme1996}
\begin{botherref}
\oauthor{\bsnm{Laflamme}, \binits{R.}},
\oauthor{\bsnm{Miquel}, \binits{C.}},
\oauthor{\bsnm{Paz}, \binits{J.P.}},
\oauthor{\bsnm{Zurek}, \binits{W.H.}}:
Perfect Quantum Error Correction Code
(1996)
\end{botherref}
\endbibitem

\bibitem[\protect\citeauthoryear{Guillaud and Mirrahimi}{2019}]{guillaud2019repetition}
\begin{barticle}
\bauthor{\bsnm{Guillaud}, \binits{J.}},
\bauthor{\bsnm{Mirrahimi}, \binits{M.}}:
\batitle{Repetition cat qubits for fault-tolerant quantum computation}.
\bjtitle{Phys. Rev. X}
\bvolume{9},
\bfpage{041053}
(\byear{2019})
\doiurl{10.1103/PhysRevX.9.041053}
\end{barticle}
\endbibitem

\bibitem[\protect\citeauthoryear{Cong et~al.}{2022}]{cong2022hardware}
\begin{barticle}
\bauthor{\bsnm{Cong}, \binits{I.}},
\bauthor{\bsnm{Levine}, \binits{H.}},
\bauthor{\bsnm{Keesling}, \binits{A.}},
\bauthor{\bsnm{Bluvstein}, \binits{D.}},
\bauthor{\bsnm{Wang}, \binits{S.-T.}},
\bauthor{\bsnm{Lukin}, \binits{M.D.}}:
\batitle{Hardware-efficient, fault-tolerant quantum computation with rydberg atoms}.
\bjtitle{Phys. Rev. X}
\bvolume{12},
\bfpage{021049}
(\byear{2022})
\doiurl{10.1103/PhysRevX.12.021049}
\end{barticle}
\endbibitem

\bibitem[\protect\citeauthoryear{Zukowski et~al.}{1993}]{zukowski1993event}
\begin{barticle}
\bauthor{\bsnm{Zukowski}, \binits{M.}},
\bauthor{\bsnm{Zeilinger}, \binits{A.}},
\bauthor{\bsnm{Horne}, \binits{M.A.}},
\bauthor{\bsnm{Ekert}, \binits{A.K.}}:
\batitle{``event-ready-detectors'' bell experiment via entanglement swapping}.
\bjtitle{Phys. Rev. Lett.}
\bvolume{71},
\bfpage{4287}--\blpage{4290}
(\byear{1993})
\doiurl{10.1103/PhysRevLett.71.4287}
\end{barticle}
\endbibitem

\bibitem[\protect\citeauthoryear{Bravyi and Kitaev}{2005}]{bravyi2005universal}
\begin{barticle}
\bauthor{\bsnm{Bravyi}, \binits{S.}},
\bauthor{\bsnm{Kitaev}, \binits{A.}}:
\batitle{Universal quantum computation with ideal clifford gates and noisy ancillas}.
\bjtitle{Phys. Rev. A}
\bvolume{71},
\bfpage{022316}
(\byear{2005})
\doiurl{10.1103/PhysRevA.71.022316}
\end{barticle}
\endbibitem

\bibitem[\protect\citeauthoryear{Cleve et~al.}{1999}]{cleve1999share}
\begin{barticle}
\bauthor{\bsnm{Cleve}, \binits{R.}},
\bauthor{\bsnm{Gottesman}, \binits{D.}},
\bauthor{\bsnm{Lo}, \binits{H.-K.}}:
\batitle{How to share a quantum secret}.
\bjtitle{Phys. Rev. Lett.}
\bvolume{83},
\bfpage{648}--\blpage{651}
(\byear{1999})
\doiurl{10.1103/PhysRevLett.83.648}
\end{barticle}
\endbibitem

\bibitem[\protect\citeauthoryear{Cormen et~al.}{2022}]{cormen2022introduction}
\begin{bbook}
\bauthor{\bsnm{Cormen}, \binits{T.H.}},
\bauthor{\bsnm{Leiserson}, \binits{C.E.}},
\bauthor{\bsnm{Rivest}, \binits{R.L.}},
\bauthor{\bsnm{Stein}, \binits{C.}}:
\bbtitle{Introduction to Algorithms}.
\bpublisher{MIT press},
\blocation{Cambridge, MA}
(\byear{2022})
\end{bbook}
\endbibitem

\bibitem[\protect\citeauthoryear{Li et~al.}{2019}]{li2019sabre}
\begin{bchapter}
\bauthor{\bsnm{Li}, \binits{G.}},
\bauthor{\bsnm{Ding}, \binits{Y.}},
\bauthor{\bsnm{Xie}, \binits{Y.}}:
\bctitle{Tackling the qubit mapping problem for nisq-era quantum devices}.
In: \bbtitle{Proceedings of the Twenty-Fourth International Conference on Architectural Support for Programming Languages and Operating Systems}.
\bsertitle{ASPLOS '19},
pp. \bfpage{1001}--\blpage{1014}.
\bpublisher{Association for Computing Machinery},
\blocation{New York, NY, USA}
(\byear{2019}).
\doiurl{10.1145/3297858.3304023}
\end{bchapter}
\endbibitem

\bibitem[\protect\citeauthoryear{Murali et~al.}{2019}]{murali2019noise}
\begin{bchapter}
\bauthor{\bsnm{Murali}, \binits{P.}},
\bauthor{\bsnm{Baker}, \binits{J.M.}},
\bauthor{\bsnm{Javadi-Abhari}, \binits{A.}},
\bauthor{\bsnm{Chong}, \binits{F.T.}},
\bauthor{\bsnm{Martonosi}, \binits{M.}}:
\bctitle{Noise-adaptive compiler mappings for noisy intermediate-scale quantum computers}.
In: \bbtitle{Proceedings of the Twenty-Fourth International Conference on Architectural Support for Programming Languages and Operating Systems}.
\bsertitle{ASPLOS '19},
pp. \bfpage{1015}--\blpage{1029}.
\bpublisher{Association for Computing Machinery},
\blocation{New York, NY, USA}
(\byear{2019}).
\doiurl{10.1145/3297858.3304075}
\end{bchapter}
\endbibitem

\bibitem[\protect\citeauthoryear{Childs et~al.}{2019}]{childs2019circuit}
\begin{bchapter}
\bauthor{\bsnm{Childs}, \binits{A.M.}},
\bauthor{\bsnm{Schoute}, \binits{E.}},
\bauthor{\bsnm{Unsal}, \binits{C.M.}}:
\bctitle{{Circuit Transformations for Quantum Architectures}}.
In: \beditor{\bsnm{Dam}, \binits{W.}},
\beditor{\bsnm{Man\v{c}inska}, \binits{L.}} (eds.)
\bbtitle{14th Conference on the Theory of Quantum Computation, Communication and Cryptography (TQC 2019)}.
\bsertitle{Leibniz International Proceedings in Informatics (LIPIcs)},
vol. \bseriesno{135},
pp. \bfpage{3}--\blpage{1324}.
\bpublisher{Schloss Dagstuhl -- Leibniz-Zentrum f{\"u}r Informatik},
\blocation{Dagstuhl, Germany}
(\byear{2019}).
\doiurl{10.4230/LIPIcs.TQC.2019.3}
\end{bchapter}
\endbibitem

\bibitem[\protect\citeauthoryear{Kitaev}{1997}]{kitaev1997quantum}
\begin{barticle}
\bauthor{\bsnm{Kitaev}, \binits{A.Y.}}:
\batitle{Quantum computations: algorithms and error correction}.
\bjtitle{Russian Mathematical Surveys}
\bvolume{52}(\bissue{6}),
\bfpage{1191}
(\byear{1997})
\doiurl{10.1070/RM1997v052n06ABEH002155}
\end{barticle}
\endbibitem

\bibitem[\protect\citeauthoryear{Tucci}{2005}]{tucci2005intro}
\begin{botherref}
\oauthor{\bsnm{Tucci}, \binits{R.R.}}:
An Introduction to Cartan's KAK Decomposition for QC Programmers
(2005).
\url{https://arxiv.org/abs/quant-ph/0507171}
\end{botherref}
\endbibitem

\bibitem[\protect\citeauthoryear{Nakajima et~al.}{2006}]{nakajima2005newalgorithm}
\begin{barticle}
\bauthor{\bsnm{Nakajima}, \binits{Y.}},
\bauthor{\bsnm{Kawano}, \binits{Y.}},
\bauthor{\bsnm{Sekigawa}, \binits{H.}}:
\batitle{A new algorithm for producing quantum circuits using kak decompositions}.
\bjtitle{Quantum Info. Comput.}
\bvolume{6}(\bissue{1}),
\bfpage{67}--\blpage{80}
(\byear{2006})
\end{barticle}
\endbibitem

\bibitem[\protect\citeauthoryear{Vidal and Dawson}{2004}]{vidal2004universal}
\begin{barticle}
\bauthor{\bsnm{Vidal}, \binits{G.}},
\bauthor{\bsnm{Dawson}, \binits{C.M.}}:
\batitle{Universal quantum circuit for two-qubit transformations with three controlled-not gates}.
\bjtitle{Phys. Rev. A}
\bvolume{69},
\bfpage{010301}
(\byear{2004})
\doiurl{10.1103/PhysRevA.69.010301}
\end{barticle}
\endbibitem

\bibitem[\protect\citeauthoryear{Bullock and Markov}{2003}]{bullock2003arbitrary}
\begin{barticle}
\bauthor{\bsnm{Bullock}, \binits{S.S.}},
\bauthor{\bsnm{Markov}, \binits{I.L.}}:
\batitle{Arbitrary two-qubit computation in 23 elementary gates}.
\bjtitle{Phys. Rev. A}
\bvolume{68},
\bfpage{012318}
(\byear{2003})
\doiurl{10.1103/PhysRevA.68.012318}
\end{barticle}
\endbibitem

\bibitem[\protect\citeauthoryear{Shende et~al.}{2006}]{shende2006synthesis}
\begin{barticle}
\bauthor{\bsnm{Shende}, \binits{V.V.}},
\bauthor{\bsnm{Bullock}, \binits{S.S.}},
\bauthor{\bsnm{Markov}, \binits{I.L.}}:
\batitle{Synthesis of quantum-logic circuits}.
\bjtitle{IEEE Transactions on Computer-Aided Design of Integrated Circuits and Systems}
\bvolume{25}(\bissue{6}),
\bfpage{1000}--\blpage{1010}
(\byear{2006})
\doiurl{10.1109/TCAD.2005.855930}
\end{barticle}
\endbibitem

\bibitem[\protect\citeauthoryear{Sousa and Ramos}{2007}]{sousa2007universal}
\begin{barticle}
\bauthor{\bsnm{Sousa}, \binits{P.B.M.}},
\bauthor{\bsnm{Ramos}, \binits{R.V.}}:
\batitle{Universal quantum circuit for n-qubit quantum gate: a programmable quantum gate}.
\bjtitle{Quantum Info. Comput.}
\bvolume{7}(\bissue{3}),
\bfpage{228}--\blpage{242}
(\byear{2007})
\end{barticle}
\endbibitem

\bibitem[\protect\citeauthoryear{Suri et~al.}{2023}]{suri2023twounitary}
\begin{barticle}
\bauthor{\bsnm{Suri}, \binits{N.}},
\bauthor{\bsnm{Barreto}, \binits{J.}},
\bauthor{\bsnm{Hadfield}, \binits{S.}},
\bauthor{\bsnm{Wiebe}, \binits{N.}},
\bauthor{\bsnm{Wudarski}, \binits{F.}},
\bauthor{\bsnm{Marshall}, \binits{J.}}:
\batitle{Two-{U}nitary {D}ecomposition {A}lgorithm and {O}pen {Q}uantum {S}ystem {S}imulation}.
\bjtitle{{Quantum}}
\bvolume{7},
\bfpage{1002}
(\byear{2023})
\doiurl{10.22331/q-2023-05-15-1002}
\end{barticle}
\endbibitem

\bibitem[\protect\citeauthoryear{Zeier et~al.}{2004}]{zeier2004gate}
\begin{barticle}
\bauthor{\bsnm{Zeier}, \binits{R.}},
\bauthor{\bsnm{Grassl}, \binits{M.}},
\bauthor{\bsnm{Beth}, \binits{T.}}:
\batitle{Gate simulation and lower bounds on the simulation time}.
\bjtitle{Phys. Rev. A}
\bvolume{70},
\bfpage{032319}
(\byear{2004})
\doiurl{10.1103/PhysRevA.70.032319}
\end{barticle}
\endbibitem

\bibitem[\protect\citeauthoryear{Steane}{1996}]{steane_code}
\begin{barticle}
\bauthor{\bsnm{Steane}, \binits{A.}}:
\batitle{Multiple-particle interference and quantum error correction}.
\bjtitle{Proceedings of the Royal Society of London. Series A: Mathematical, Physical and Engineering Sciences}
\bvolume{452}(\bissue{1954}),
\bfpage{2551}--\blpage{2577}
(\byear{1996})
\doiurl{10.1098/rspa.1996.0136}
\end{barticle}
\endbibitem

\bibitem[\protect\citeauthoryear{Aaronson and Gottesman}{2004}]{aaronson2004improved}
\begin{barticle}
\bauthor{\bsnm{Aaronson}, \binits{S.}},
\bauthor{\bsnm{Gottesman}, \binits{D.}}:
\batitle{Improved simulation of stabilizer circuits}.
\bjtitle{Phys. Rev. A}
\bvolume{70},
\bfpage{052328}
(\byear{2004})
\doiurl{10.1103/PhysRevA.70.052328}
\end{barticle}
\endbibitem

\bibitem[\protect\citeauthoryear{Eastin and Knill}{2009}]{eastin2009restrictions}
\begin{barticle}
\bauthor{\bsnm{Eastin}, \binits{B.}},
\bauthor{\bsnm{Knill}, \binits{E.}}:
\batitle{Restrictions on transversal encoded quantum gate sets}.
\bjtitle{Phys. Rev. Lett.}
\bvolume{102},
\bfpage{110502}
(\byear{2009})
\doiurl{10.1103/PhysRevLett.102.110502}
\end{barticle}
\endbibitem

\end{thebibliography}

\end{document}